\documentclass[amsthm]{elsart}

\usepackage{yjsco}
\usepackage{mathptmx}

\usepackage{amsthm}
\usepackage{amsmath}
\usepackage{mathbbol}
\usepackage{amssymb}
\usepackage{graphicx}
\usepackage{algorithmic}
\usepackage{url}
\usepackage{epsfig}
\usepackage{wrapfig}
\usepackage{algorithm}

\newcommand{\donotshow}[1]{}

\newcommand{\ignore}[1]{}

\providecommand{\qed}{\rule[-0.2ex]{0.3em}{1.4ex}}

\newcommand{\assign}{\mathbin{:=}}

\providecommand{\R}{\mathbb{R}}
\newcommand{\N}{\mathbb{N}}
\newcommand{\Z}{\mathbb{Z}}
\newcommand{\Q}{\mathbb{Q}}
\newcommand{\C}{\mathbb{C}}

\newlength{\setspacing}
\providecommand{\sset}[1]{\{\, #1 \,\}}

\newcommand{\mbegin}{\{\ \ }
\newcommand{\mend}{\}}

\newlength{\mleftindent}
\newlength{\mindent}
\newlength{\mboxwidth}
\newcommand{\mincrement}{\addtolength{\mboxwidth}{-\mindent}}
\newcommand{\mdecrement}{\addtolength{\mboxwidth}{\mindent}}

\newlength{\preprogramskip}
\newlength{\postprogramskip}
\newlength{\mexpwidth}
\newlength{\mexpindent}
\newcommand{\indentafterkeyword}{\hspace*{0.5em}}

\newcommand{\mslifelse}[3]  
{\setlength{\mexpwidth}{\mboxwidth}%
\settowidth{\mexpindent}{{\bf if\indentafterkeyword}}%
\addtolength{\mexpwidth}{-\mexpindent}%
{\bf if\indentafterkeyword}\parbox[t]{\mexpwidth}{#1}\\
\mincrement \mbegin \parbox[t]{\mboxwidth}{#2 \mend} \mdecrement \\
{\bf else} \\
\mincrement \mbegin \parbox[t]{\mboxwidth}{#3}\\
\mend \mdecrement
}

{\vspace{-0.3em}\begin{itemize}
\setlength{\itemsep}{0.2\itemsep}%
\setlength{\parskip}{0.2\parskip}%
\setlength{\topsep}{0.0\topsep}%
}
{\end{itemize}}

{\begin{itemize}
\setlength{\itemsep}{0.2\itemsep}%
\setlength{\parskip}{0.2\parskip}%
\setlength{\topsep}{0.0\topsep}%
}
{\end{itemize}}

{\vspace{-0.3em}\begin{description}
\setlength{\itemsep}{0.2\itemsep}%
\setlength{\parskip}{0.2\parskip}%
\setlength{\topsep}{0.0\topsep}%
}
{\end{description}}

{\begin{description}
\setlength{\itemsep}{0.2\itemsep}%
\setlength{\parskip}{0.2\parskip}%
\setlength{\topsep}{0.0\topsep}%
}
{\end{description}}

{\vspace{-0.3em}\begin{enumerate}
\setlength{\itemsep}{0.2\itemsep}%
\setlength{\parskip}{0.2\parskip}%
\setlength{\topsep}{0.0\topsep}%
}
{\end{enumerate}}

{\begin{enumerate}
\setlength{\itemsep}{0.2\itemsep}%
\setlength{\parskip}{0.2\parskip}%
\setlength{\topsep}{0.0\topsep}%
}
{\end{enumerate}}

\newlength{\proofpostskipamount}\newlength{\proofpreskipamount}
               {\par\vspace{0.5ex}\noindent{\bf Proof #1:}\hspace{0.5em}}%
               {\nopagebreak%
                \strut\nopagebreak%
                \hspace{\fill}\qed\par\medskip\noindent}

\newlength{\mydefwidth}
\newlength{\mytextwidth}

\newcommand{\myurl}[1]{{\footnotesize \url{#1}}}

\usepackage{url} 
\urldef{\mails}\path|msagralo@mpi-inf.mpg.de|

\newcommand{\var}{\operatorname{var}}
\newcommand{\Otilde}{\tilde{O}}

\newcommand{\sgn}{\mathrm{sign}}

\newcommand{\rev}{\operatorname{rev}}

\newcommand{\IBox}{\mathfrak{B}}

\newcommand{\T}{\mathcal{T}}

\newcommand{\dcm}{\textsc{Dcm}}

\numberwithin{equation}{section}
\numberwithin{figure}{section}

\begin{document}

\begin{frontmatter}

\title{On the Complexity of Real Root Isolation}


\author{Michael Sagraloff}
\address{MPI for Informatics, Saarbr\"ucken, Germany}
\ead{msagralo@mpi-inf.mpg.de}

\begin{abstract}
We introduce a new approach to isolate the real roots of a square-free polynomial $F=\sum_{i=0}^n A_i x^i$ with real coefficients. It is assumed that each coefficient of $F$ can be approximated to any specified error bound. The presented method is exact, complete and deterministic. Due to its similarities to the Descartes method, we also consider it practical and easy to implement. Compared to previous approaches, our new method achieves a significantly better bit complexity. It is further shown that the hardness of isolating the real roots of $F$ is exclusively determined by the geometry of the roots and not by the complexity or the size of the coefficients. For the special case where $F$ has integer coefficients of maximal bitsize $\tau$, our bound on the bit complexity writes as $\Otilde(n^3\tau^2)$ which improves the best bounds known for existing practical algorithms by a factor of $n=\deg F$.\\ 
The crucial idea underlying the new approach is to run an approximate version of the Descartes method, where, in each subdivision step, we only consider approximations of the intermediate results to a certain precision. We give an upper bound on the maximal precision that is needed for isolating the roots of $F$. For integer polynomials, this bound is by a factor $n$ lower than that of the precision needed when using exact arithmetic explaining the improved bound on the bit complexity. 
\end{abstract}


\end{frontmatter}

\section{Introduction}\label{intro}

Finding the roots of a univariate polynomial $F\in\R[x]$ can be considered as the fundamental problem of computational algebra, and there exist numerous approaches dedicated to approximate the real roots of $F$. We mainly distinguish between purely numerical methods such as Newton iteration and exact and complete methods such as those based on Descartes' Rule of Signs or Sturm Sequences. The latter approaches apply to polynomials with rational coefficients and guarantee to compute a set of disjoint \emph{isolating} intervals. That is, each of these intervals contains exactly one root and the union of all intervals covers all real roots of $F$. In this paper, we propose an algorithm which extends the Descartes method to arbitrary square-free polynomials with real coefficient. 
Throughout the paper, 
\begin{align}
F(x):=\sum_{i=0}^n A_i x^{i}\in\R[x] \label{polyF}
\end{align}
denotes a square-free polynomial of degree $n\ge 2$ with real coefficients $A_{i}$, where $|A_{n}|\ge 1$. We define $\tau$ to be the minimal positive integer with $\max_{i=0,\ldots,n-1} \frac{|A_{i}|}{|A_n|}<2^{\tau}$. It is assumed that each coefficient $A_{i}$ can be approximated to any specified precision and we refer to such coefficients as \emph{bitstream coefficients}. The roots of $F$ are denoted by $\xi_{1},\ldots,\xi_{n}\in\C$ and $\Gamma_F:=\log(\max_i|\xi_i|)$ denotes the corresponding logarithmic root bound. The \emph{separation} $\sigma_i:=\sigma(\xi_{i},F)$ of $\xi_{i}$ is defined as the minimal distance of $\xi_{i}$ to any root $\xi_{j}\neq \xi_{i}$, the separation $\sigma_F$ of $F$ is defined as the minimum of all $\sigma(\xi_i,f)$, and $\Sigma_F:=-\sum_{i=1}^{n} \log \sigma_i.$\\

\subsection{Main results and related work} 
We present an \emph{exact} and \emph{deterministic} algorithm which computes isolating intervals $I_{1},\ldots,I_{m}$ for the real roots of $F$. We further provide a detailed complexity analysis showing that our algorithm needs no more than 
\begin{align}
\Otilde(n(\Sigma_F+n\Gamma_F)^2)=\Otilde(n(\Sigma_F+n\tau)^2)\label{result:complexity}
\end{align}
bit operations\footnote{$\Otilde$ indicates that we omit polylogarithmic factors} and demands for approximations of the coefficients of $F$ to $\tilde{O}(\Sigma_F+n\Gamma_F)$ bits after the binary point. Our results show that the complexity of isolating the real roots does not depend on whether the given polynomial has irrational, rational or integer coefficients. In fact, the hardness of isolating the roots of $F$ is exclusively determined by the degree of $F$ and the quantities $\Gamma_F$ and $\Sigma_F$ which only depend on the location of the roots of $F$.
For a polynomial $F$ with integer coefficients, the bound in (\ref{result:complexity}) writes as $\Otilde(n^3\tau^2)$ which improves the best bounds known for other practical methods such as the Descartes method~\cite{Alesina-Galuzzi,collins-akritas:76,eigenwillig-sharma-yap:descartes:06,mrr:bernstein:05,rouillier-zimmermann:roots:04}, Sturm's method~\cite{du-sharma-yap:sturm:07,lickteig-roy:sequences:01} or the continued fraction method~\cite{Akritas:1980:FEA,Sharma,tsigaridasE08,Vincent} by a factor of $n$. To the best of our knowledge, this is the first time where it is shown that approximation leads to a better worst case complexity for real root isolation, a fact which has already been observed in experiments~\cite{snc-benchmarks09,rouillier-zimmermann:roots:04}. We consider this new result as an important step to further reduce the gap (with respect to worst case bit complexity) between practical and efficient algorithms for real root isolation and asymptotically fast methods for isolating all complex roots as proposed by Sch\"onhage \cite{schonhage:fundamental} and Pan~\cite{pan:seq-parallel:87,pan:history-progress:97} in the eighties and nineties. The latter methods achieve almost optimal complexity bounds $\tilde{O}(n^3\tau)$ for the benchmark problem of isolating all complex roots but both methods lack evidence of being efficient in practice; see~\cite{Gourdon96} for an implementation of the splitting circle method within the Computer Algebra system Pari/GP. Due to its similarities to the Descartes method, we consider the proposed algorithm practical and easy to implement. The latter claim has already been proven by means of a recent implementation from A. Strzebonski and E. Tsigaridas~\cite{st:polyalgebraic:11} ``in C as part of the core library of \textsc{Mathematica}.''\\

\begin{figure}[t]
\begin{center}
\vspace{-5.6cm}\includegraphics[width=12cm]{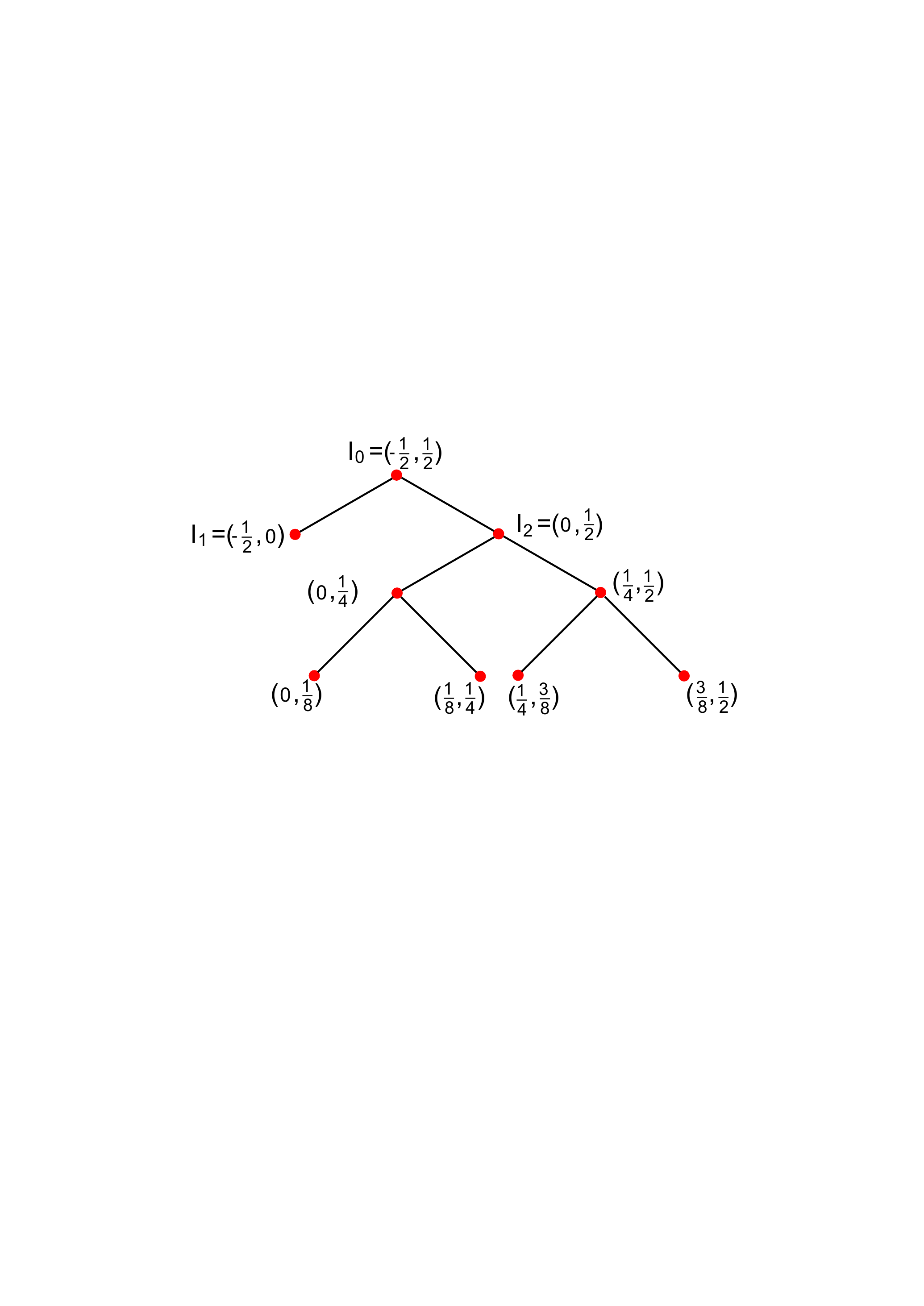}\end{center}
\vspace{-7.3cm}\caption{\label{fig:desctree} \footnotesize{The above figure shows the recursion tree induced by the Descartes method when applied to the polynomial $f(x):=16\sqrt{2}x^2-8x+\frac{\pi}{4}$ (with roots $z_1=0.06\ldots$ and $z_2=0.29\ldots$). For each interval $I=(a,b)$ in the subdivision process, we have to compute $f_I(x)=f(a+(b-a)x)$. For instance, for $I=(\frac{1}{4},\frac{1}{2})$, we have $f_I(x)=f(\frac{1}{4}+\frac{x}{4})=\sqrt{2}x^2+(2\sqrt{2}-2)x+\sqrt{2}+\frac{\pi}{8}$. The Bitstream solvers presented in~\cite{ms-detbitdesc-09,s-bitstream-mcs11} initially start with an approximation $g$ of $f$ to a certain number of bits; e.g., $g(x)=\frac{11585}{512} x^2-8x+\frac{201}{512}$ approximates $f$ to $10$ bits after the binary point. Then, the Descartes method is applied to $g$, that is, for each interval $I=(a,b)$, $g_I(x)=g(a+(b-a)x)$ is computed; e.g., $g(\frac{1}{4}+\frac{x}{4})=\frac{11585}{8192} x^2+\frac{3393}{4096}x-\frac{1583}{8192}$. Given that $g$ is a sufficiently good approximation of $f$, it is shown that the roots of $f$ can be isolated in this way. 
Our new approach follows a similar strategy, that is, we start with an approximation $\tilde{f}_{I_0}$ of $f_{I_{0}}$ to a certain number $\rho_{I_0}=\rho$ of bits. Then, we recursively compute approximations $\tilde{f}_I$ of $f_I$ to $\rho_I$ bits, where $\rho_I$ is updated in each step. In contrast to the previous method, the polynomials $\tilde{f}_{I}$ do not necessarily correspond to a specific initial approximation $g$ of $f$.
We illustrate this by means of the above example: We start with 
$\tilde{f}_{I_0}(x)=\frac{11585}{512}x^2-\frac{31363}{1024}+\frac{5145}{512}$ 
which approximates 
$f_{I_0}(x)=f(-\frac{1}{2}+x)=16\sqrt{2}x^2-16\sqrt{2}x+4-8x+4+\frac{\pi}{8}+4\sqrt{2}
$ to $\rho_{I_0}=10$ bits. Then, $\tilde{f}_{I_0}(\frac{x}{2})$ and 
$\tilde{f}_{I_0}(\frac{1}{2}+\frac{x}{2})$ are evaluated and the result is rounded to $9$ bits 
after the binary point. The resulting polynomials are then approximations of 
$f_{I_{1}}(x)=f(\frac{1}{2}+\frac{x}{2})$ and $f_{I_{2}}(x)=f(\frac{x}{2})$ to 
$\rho_{I_1}=\rho_{I_2}=8$ bits, respectively (see 
Lemma~\ref{lem:taylorshift}). In the following bisection steps, we proceed in 
exactly the same manner. For instance, for the interval $I=(\frac{1}{4},\frac{1}{2})$, we 
obtain $\tilde{f}_I(x)=\frac{181}{128} x^2+\frac{53}{64} x-\frac{25}{128}$ 
which approximates $f_I$ to $\rho_I=6$ bits after the binary point.
}}\vspace{0.25cm}
\end{figure}

The crucial idea underlying the presented method is to use an ``approximate version" of the Descartes method. More precisely, we first consider a scaled polynomial $f(x):=F(2^{\Gamma+1} x)/A_n$\label{deftau}, where $\Gamma$ is an integer approximation of $\Gamma_F$ with $\Gamma_F\le\Gamma\le \Gamma_F+4\log n$; see Section~\ref{sec:scaling} and Appendix~\ref{sec:rootbound}. Then, all roots of $f$ are contained within the disc of radius $1/2$ centered at the origin. 
In a second step, we apply a modified Descartes method to isolate the roots of $f$. However, instead of computing the exact intermediate results obtained in the subdivision process, we only consider approximations to a certain number of bits. Whereas other methods~\cite{cjk-iaicad-02,BitstreamDescartes,Johnson_Krandick:97,mrr:bernstein:05,rouillier-zimmermann:roots:04} proceed in a similar way by using interval polynomials, our new method considers a specific approximation in each step and updates the possible approximation error. In~\cite{ms-detbitdesc-09,s-bitstream-mcs11}, a similar approach was proposed. Therein, the proposed algorithms also initially start with an approximation $g$ of $f$, however, all intermediate results correspond to the initial approximation $g$ and are computed exactly. In contrast, we propose to consider independent approximations of the intermediate results at \emph{each node} of the recursion tree; see Figure~\ref{fig:desctree} for a more detailed example.\\

How is it possible that, for integer polynomials, an approximate version of the Descartes method is more efficient than the original ``exact version"? Let us first consider the ``exact Descartes method": Its complexity analysis shows that, for each interval (node) $I=(a,b)$ in the recursion tree, the dominating costs are those for the computation of the Taylor expansion $f_{I}(x):=f(a+(b-a)x)$ at $a$; see Section~\ref{descartes} for a more comprehensive treatment. In each bisection step, the polynomials $f_{I_l}$ and $f_{I_r}$ (corresponding to the left and the right subinterval of $I$) are recursively computed from $f_I$ by replacing $x$ by $x/2$, followed by a Taylor shift by $1$, that is, $x\mapsto x+1$. More precisely, we have $f_{I_l}(x)=f_I(x/2)$ and $f_{I_r}(x)=f_{I_l}(x+1)$. In each iteration, the bitsize of the coefficients of $f_I$ increases by $n$ bits, and since the recursion tree has depth bounded by $h_{\max}=\Otilde(n\tau)$, the representation of $f_{I}$ eventually demands for at most $\tau+nh_{\max}=\Otilde(n^2\tau)$ bits. Hence, assuming asymptotically fast Taylor shift~\cite{Ger04,GG97}, the computation of a certain $f_{I}$ amounts for $\Otilde(n^{3}\tau)$ bit operations.

Now, let us turn to the approximate method: In Section~\ref{sec:apx}, we show that, for an arbitrary approximation $g$ of $f$ to $\rho_{\max}=\tilde{O}(n\tau)$ bits after the binary point, corresponding roots of $f$ and $g$ are almost at the same location with respect to their separations; see Theorem~\ref{lem:precision} and Appendix~\ref{integer} for a more precise result. Thus, for each interval $I$, it should suffice to consider approximations $\tilde{f}_I$ of $f_I$ to $\rho_{\max}$ bits after the binary point. Starting with an approximation of $f$ to $\rho_{\max}+2h_{\max}=\tilde{O}(n\tau)$ bits after the binary point, we can iteratively obtain such approximations $\tilde{f_I}$. Namely, $\tilde{f}_I$ can be recursively computed such that the approximation error quadruples at most in each bisection step and the height of the recursion tree is bounded by $h_{\max}$.
Eventually, all polynomials $\tilde{f_I}$ are represented by $\tilde{O}(n\tau)$ bits (instead of $\Otilde(n^2\tau)$ bits for the exact counterpart $f_I$) and, thus, the cost at each node decreases by a factor $n$.\\

We will prove the above result for the more general setting where $F$ is a polynomial with arbitrary real coefficients. More precisely, we show that it suffices to approximate each $f_I$ to a number of bits after the binary point bounded by $O(\Sigma_f+n)=\tilde{O}(\Sigma_F+n\Gamma_F)$. Then, each $\tilde{f_I}$ is represented by $\tilde{O}(\Sigma_F+\tau+n\Gamma_F)$ bits and, as a consequence, the cost at each node is bounded by $\Otilde(n(\Sigma_F+\tau+n\Gamma_F))$ bit operations. We remark that, due to Appendix~\ref{sec:apx}, we have $\tau=\tilde{O}(n\Gamma_F)$ and, thus, the latter bound writes as $\Otilde(n(\Sigma_F+n\Gamma_F))$.
The additional factor $\Sigma_F+n\Gamma_F$ in the bound (\ref{result:complexity}) on the bit complexity is due to the size of the induced recursion tree.

\subsection{Outline}
In Section~\ref{basics}, we first introduce some basic notations. Furthermore, we derive a bound on how good $f$ has to be approximated such that 
its roots stay at almost the same place with respect to the corresponding 
separations. Eventually, we revise the Descartes method before presenting our slight modification $\dcm$ of it in Section~\ref{modifieddescartes}. In Section~\ref{algorithm}, we present our new algorithm to isolate the roots of $F$ and provide the corresponding complexity bounds. We conclude in Section~\ref{conclusion}. Parts of the complexity analysis as well as pseudo-code for our subroutines is outsourced to the Appendix.

\section{Preliminaries}\label{basics}

\subsection{Some Notations} 
For an interval $I=(a,b)$, $w(I):=b-a$ denotes the \emph{width}, $m(I):=\frac{a+b}{2}$ the \emph{center}, and $r(I)=\frac{w(I)}{2}$ the \emph{radius} of $I$. Furthermore,
$$I^+=(a^+,b^+):=(a-\frac{w(I)}{4n},b+\frac{w(I)}{4n})\quad\text{and}\quad \tilde{I}=(\tilde{a},\tilde{b}):=(a-\frac{w(I)}{2n},b+\frac{w(I)}{2n})$$
denote extensions of $I$ by $\frac{w(I)}{4n}$ and $\frac{w(I)}{2n}$ (to both sides), respectively. We will need these intervals for our modified version of the Descartes method as presented in Section~\ref{modifieddescartes}.
An (open) disc in $\C$ is denoted by $\Delta=\Delta_r(m)$, where $m\in\C$ indicates the center of $\Delta$ and $r\in\R^+$ its radius. The closure of a disc $\Delta$ or an interval $I$ is denoted by $\bar{\Delta}$ and $\bar{I}$, respectively.

\subsection{Scaling the Polynomial}\label{sec:scaling}

Instead of isolating the roots of the given polynomial $F$ as in (\ref{polyF}), we consider the equivalent task of isolating the roots of a "scaled'' polynomial $f$ which is defined as follows: We first compute an integer approximation $\Gamma\in\mathbb{N}$ of the exact logarithmic root bound $\Gamma_F=\log(\max_i|\xi_i|)$ of $F$ such that 
\begin{align}
\Gamma_F\le\Gamma<4\log n+\Gamma_F.\label{Gamma}
\end{align}
This computation can be done with $\tilde{O}((n\Gamma_F)^2)$ bit operations and demands for an approximation of $F$ to $\tilde{O}(n\Gamma_F)$ bits after the binary point; see Appendix~\ref{sec:rootbound}. We can further assume that $\Gamma\le \tau+1$ due to Cauchy's Bound~\cite{yap-fundamental} $B_{CB}:=1+\max_i\frac{|A_i|}{|A_n|}<1+2^{\tau}$ on the modulus of all roots.
Now, we define
\begin{align}
f(x)=\sum_{i=0}^{n}a_{i}x^{i}:=\frac{F(2^{\Gamma+1}\cdot x)}{A_{n}}.\label{polyf}
\end{align}
It follows that all roots $z_{1}=\xi_{1}\cdot 2^{-(\Gamma-1)},\ldots,z_{n}=\xi_{n}\cdot 2^{-(\Gamma-1)}$ of $f$ are contained within the disc $\Delta_{1/2}(0)$ and the absolute value of each coefficient $a_i$ of $f$ is bounded by $2^{n(\Gamma+1)+\tau}=2^{O(n\tau)}$. In practice, it might be worth to investigate in an even tighter root bound $\Gamma$ as described in~\cite[Section 2.4]{eigenwillig:thesis} in order to prevent the coefficients of $f$ to become unnecessarily large.
We further remark that the separations of corresponding roots of $F$ and $f$ scale by $2^{\Gamma+1}$ (i.e., $\sigma(\xi_{i},F)=2^{\Gamma+1}\cdot\sigma(z_{i},f)$). Thus, \begin{align}
\Sigma_f=-\sum_{i=1}^{n}\log\sigma(z_{i},f)= \Sigma_F+n(\Gamma+1)=O(n\tau+\Sigma_F).\label{Sigmaf}
\end{align}

\subsection{Approximating Polynomials}\label{sec:apx}

We assume that the coefficients of $F$ are given as infinite bitstreams, that is, for a given $\rho\in\N$, we can ask for an approximation of $F$ to $\rho$ bits after the binary point. More precisely, each coefficient $A_i$ is approximated by a binary fraction 
$\tilde{A}_i = m_i\cdot 2^{-\rho}$ with $m_i \in \Z$ and $|A_i - \tilde{A}_i| \le
2^{-\rho}$, e.g., $\tilde{A}_{i} = \sgn(A_i)\lfloor |A_{i} 2^{\rho}| \rfloor 2^{-\rho}$. We call a polynomial $\tilde{F}\in \Q[x]$ obtained in this way a \emph{$\rho$-binary approximation of $F$}. We remark that, in order to get a $\rho$-binary approximation of $f$, it suffices to approximate $F$ to $n(\Gamma+1)+\rho+\tau+1$ bits after the binary point. Namely, given approximations $\tilde{A}_i=A_i+\mu_i$, with $|\mu_i|\le\mu:=2^{-(n(\Gamma+1)+\rho+\tau+1)}$ for all $i$, it follows that $$\left|\frac{\tilde{A_i}}{\tilde{A}_n}-\frac{A_i}{A_n}\right|=\frac{|\mu_i A_n+\mu_n A_i|}{|\tilde{A}_nA_n|}\le \mu\cdot \frac{|A_n|+|A_i|}{|A_n\tilde{A}_n|}< \mu(1+2^\tau)<2^{-(n(\Gamma+1)+\rho)}.$$
Thus, $\tilde{a}_i:=\frac{\tilde{A_i}}{\tilde{A}_n}(2\Gamma)^i$ approximates $a_i=\frac{A_i}{A_n}(2^{\Gamma+1})^i$ to an error less than $2^{-\rho}$.

For an arbitrary polynomial $g(x):=\sum_{i=0}^{m}g_{i}x^{i}\in\C[x]$ with complex coefficients and an arbitrary non-negative real number $\mu\in\R^+_0$, we define
\begin{align*}
[g]_{\mu}:=\left\{\tilde{g}(x)=\sum_{i=0}^{n}\tilde{g}_{i}x^{i}\in\C[x]:|g_{i}-\tilde{g}_{i}|\le \mu\text{ for all }i=0,\ldots,n\right\}
\end{align*}
\emph{the set of all $\mu$-approximations of $g$.} We remark that, since the coefficients of modulus less than $\mu$ can be approximated by zero, a $\mu$-approximation $\tilde{g}$ of $g$ might have lower degree than $g$.\\

\noindent\textbf{Example.} For $g(x):=\frac{12256}{65589}x^{10}-2x^2+\frac{1}{243}x-\frac{9}{16}$, the polynomial $\tilde{g}(x):=\frac{11}{64}x^{10}-2x^2-\frac{9}{16}$ constitutes a $6$-binary approximation and $\tilde{g}(x):=-2x^2-\frac{3}{4}$ a $2$-binary approximation of $g$.\\

\subsection{Taylor Shifts}\label{taylor}

For an arbitrary polynomial $g\in\C[x]$ and arbitrary values $m\in\C$, $\lambda\in\R\backslash \{0\}$, let \begin{align}
g_{[m,\lambda]}(x):=g(m+\lambda x).\label{tshift}
\end{align}

The following lemma provides error bounds on how the absolute approximation error $\mu$ of a polynomial $\tilde{g}\in [g]_{\mu}$ scales under the transformation $x\mapsto m+\lambda x$:

\begin{lem}
\label{lem:taylorshift}
For $\mu\in\R_0^+$ and $\tilde{g}\in [g]_{\mu}$ an arbitrary $\mu$-approximation of a polynomial $g\in\C[x]$ of degree $n$, it holds that
\begin{itemize}
\item[(i)] $\tilde{g}_{[\frac{1}{2},\frac{1}{2}]}\in 
[g_{[\frac{1}{2},\frac{1}{2}]}]_{2\mu}$,
\item[(ii)] $\tilde{g}_{[-\frac{1}{4n},1+\frac{1}{2n}]}\in [g_{[-\frac{1}{4n},1+\frac{1}{2n}]}]_{4\mu}$, 
\item[(iii)] $\tilde{g}_{[-\frac{1}{2},1]}\in [g_{[-\frac{1}{2},1]}]_{2^n\mu}$, and $\tilde{g}_{[1,1]}\in [g_{[1,1]}]_{2^n\mu}$.
\end{itemize}
\end{lem}

\begin{proof}
For $h(x):=(g-\tilde{g})(x)=\mu_{n} x^{n}+\ldots+\mu_1x+\mu_0$, the absolute value of each coefficient $\mu_i$ is bounded by $\mu$. Let $m\in\C$ and $\lambda\in\R\backslash \{0\}$ be arbitrary values, then
\begin{align}
h(m+\lambda x) & =\sum_{i=0}^n\mu_i (m+\lambda x)^i=\sum_{i=0}^n\mu_i\sum_{k=0}^i x^k\lambda^k m^{i-k}\binom{i}{k} = \sum_{k=0}^n x^k\sum_{i=k}^n \mu_i m^{i-k}\lambda^k \binom{i}{k}  \label{computation:shift}
\end{align}
Thus, for $|m|<1$, the absolute value of the coefficient of $x^k$ is bounded by
\begin{align}
\mu |\lambda|^k\cdot \sum_{i\geq k} |m|^{i-k}\binom{i}{k}=\mu|\lambda|^k\cdot\sum_{i\geq 0}|m|^i\binom{k+i}{k}=\mu |\lambda|^k\cdot\frac{1}{(1-|m|)^{k+1}}, \label{computation:shift2}
\end{align}
where we used 
\begin{align*}
(1-|m|)^{-(k+1)}=\sum_{i\geq 0}\binom{-(k+1)}{i}(-1)^i |m|^{i}=\sum_{i\geq 0}\binom{k+i}{i} |m|^{i}=\sum_{i\geq 0}\binom{k+i}{k} |m|^{i}.
\end{align*}
For $m=\lambda=1/2$, it follows that all coefficients of $h$ are bounded by $2\mu$. This shows (i). For $m=-\frac{1}{4n}$ and $\lambda=1+\frac{1}{2n}$, (\ref{computation:shift2}) implies that
\[\tilde{g}_{[-\frac{1}{4n},1+\frac{1}{2n}]}\in [g_{[-\frac{1}{4n},1+\frac{1}{2n}]}]_{\mu\frac{8}{7}\cdot \left(\frac{1+1/(2n)}{1-1/(4n)}\right)^n}\subset [g_{[-\frac{1}{4n},1+\frac{1}{2n}]}]_{4\mu}\]
because $\frac{8}{7}\cdot\left(\frac{1+1/(2n)}{1-1/(4n)}\right)^n\le \frac{8^3}{7^3}\cdot \sqrt{e}\le 4$. Hence, (ii) follows. The first part of (iii) is also a direct implication of (\ref{computation:shift2}).
The second claim in (iii) follows
from the computation in (\ref{computation:shift}) since each $\mu_{i}$ is then ($m=\lambda=1$) bounded by 
$\mu\cdot\sum_{i=k}^{n}\binom{i}{k}=\sum_{i=k}^n\binom{i}{i-k}=\sum_{i=0}^{n-k}\binom{i+k}{i}\le \sum_{i=0}^{n-k}\binom{n}{i}\le 2^{n}\cdot\mu$.
\end{proof}

\subsection{On Sufficiently Good Approximation}\label{approximation}

In the next step, we derive a bound on how good $f$ has to be approximated by an $\tilde{f}$ such that, for all $i$, the distance of corresponding roots $z_{i}$ and $\tilde{z}_{i}$ of $f$ and $\tilde{f}$ is small with respect to the separation $\sigma(z_{i},f)$. The following considerations are mainly adopted from our studies in~\cite{s-bitstream-mcs11}. Only for the sake of comprehensibility, we decided to integrate the results in this paper as well.
We start with the following definition:

\begin{defn}\label{def:sufficintlylarge}
Let $t\ge 1$ be an arbitrary real value and $f$ a polynomial as in (\ref{polyf}). We define 
\begin{align}\
\mu(f,t):=\frac{1}{t}\cdot\min_{i=1,\ldots,n}\left|\frac{\sigma(z_i,f) 
f'(z_i)}{8n^2}\right| \label{mugt}
\end{align}
We call a $\rho\in\mathbb{N}$ \emph{sufficiently large\footnote{This definition is motivated by our results in Theorem~\ref{lem:precision} and Section~\ref{dcmL}} with respect to $f$} if 
\begin{align}
\rho\ge \rho_f:=\lceil-\log\mu(f,64n^2)\rceil=O(\Sigma_f+\log n-\log|a_n|)=O(\Sigma_F+\log n).\label{sufficientlylarge}
\end{align}
\end{defn}

The upper bound for $\rho_f$ in (\ref{sufficientlylarge}) follows from
$$\sigma(z_i,f)\cdot|f'(z_i)|=\sigma(z_i,f)\cdot|a_n|\prod_{j\neq i}|z_i-z_j|\ge \sigma(z_i,f)\cdot|a_n|\prod_{j\neq i} \sigma(z_j,f)=|a_n|2^{-\Sigma_f}.$$
and
$
\Sigma_f-\log |a_n|=\Sigma_F+n(\Gamma+1)-\log(2^{n(\Gamma+1)})=\Sigma_F
$. The following theorem gives an answer to our question raised above:

\begin{thm}
\label{lem:precision}
Let $f$ be the polynomial as defined in (\ref{polyf}), $t\ge 1$ and $\tilde{f}\in [f]_{\mu(f,t)}$.
\begin{itemize}
\item[(i)] For all $i=1,\ldots,n$, the disc $\Delta_i:=\Delta_{\sigma(z_i,f)/(tn)}$ contains the root $z_i$ of $f$ and a corresponding counterpart $\tilde{z}_i$ of $\tilde{f}$.
\item[(ii)] For each $z\in\C\backslash\bigcup_{i=1}^n \Delta_i$, it holds that $|f(z)|>(n+1)\mu(f,t).$
\end{itemize}
\end{thm}
\begin{proof}
Since all roots of $f$ are contained within $\Delta_{1/2}(0)$, it follows that $\sigma(z_{i},f)<1$ for all $i$ and, thus, each disc $\Delta_i$ is completely contained within the unit disc. For an arbitrary point $z\in\partial\Delta_i$ on the boundary of $\Delta_i$, we have
\begin{align}
|f(z)|&=|a_n|\prod_{j=1}^n|z-z_j|= \frac{\sigma(z_i,f)}{tn} \left(\prod_{1\le j\le n, j\neq i}\left|\frac{z-z_j}{z_i-z_j}\right|\right)\cdot |a_n|\left(\prod_{1\le j\le n, j\neq i}|z_i-z_j|\right)\nonumber \\ 
& = \frac{\sigma(z_i,f)|f'(z_i)|}{tn}\prod_{1\le j\le n, j\neq i}\left|\frac{z-z_j}{z_i-z_j}\right|\geq\frac{\sigma(z_i,f)|f'(z_i)|}{tn}\prod_{1\le j\le n, j\neq i}\frac{|z_i-z_j|-|z-z_i|}{|z_i-z_j|}\nonumber \\
&\ge \frac{\sigma(z_i,f)|f'(z_i)|}{tn}\left(1-\frac{1}{tn}\right)^{n-1}
>\frac{\sigma(z_i,f)|f'(z_i)|}{2.72\cdot tn}> (n+1)\mu(f,t).\nonumber
\end{align}
In addition, since $\tilde{f}\in [f]_{\mu(f,t)}$ and $|z|<1$, we have $|(f-\tilde{f})(z)|< (n+1)\mu(f,t)<|f(z)|$. Hence, (i) follows from Rouch\'e's Theorem applied to the discs $\Delta_i$ and the functions $f$ and $\tilde{f}$. 
For (ii), we remark that $f$ is a holomorphic function on $\C\backslash\bigcup_{i=1}^n \Delta_i$ and, thus, $|f(z)|$ becomes minimal for a point $z$ on the boundary of one of the discs $\Delta_i$. 
\end{proof}

From the last theorem, it follows that, for given $f$ as in (\ref{polyf}), it suffices to approximate the coefficients of $f$ to $\rho=O(\Sigma_f+\log n-\log|a_n|)=O(\Sigma_F+\log n)$ bits after the binary point to guarantee that each approximation $\tilde{f}\in [f]_{2^{-\rho}}$ has its roots at almost the same location as $f$.

\begin{cor}\label{cor:precision}
Let $f$ be a polynomial as defined in (\ref{polyf}) and $\rho\in\mathbb{N}$ be sufficiently large with respect to $f$, that is, $\rho\ge \rho_f$ with $\rho_f$ as defined in (\ref{sufficientlylarge}). Then, each root $z_i$ moves by at most $\frac{\sigma(z_i,f)}{64n^3}$ when passing from $f$ to an arbitrary approximation $\tilde{f}\in [f]_{2^{-\rho}}$. In particular, real roots of $f$ stay real and non-real roots stay non-real. Furthermore, for any $z\in\C$ with $|z-z_i|\ge \frac{\sigma(z_i,f)}{64n^3}$ for all $i$, it holds that $|f(z)|>(n+1)2^{-\rho_f}.$
\end{cor}

\subsection{The Descartes Method}\label{descartes}

We first resume some basic facts about the Descartes method for isolating the real roots of a polynomial $f(x) = \sum_{i=0}^n a_i x^n\in\R[x]$. Descartes' Rule of Signs states that the number $\var(f)$ of sign changes in the coefficient sequence of $f$,
that is, the number of pairs $(i,j)$ with $i < j$, $a_ia_j < 0$, and $a_{i+1} =
\ldots = a_{j-1} = 0$, is not smaller than and of the same parity as the number
of positive real roots of $f$. If
$\var(f) = 0$, then $f$ has no positive real root, and if $\var(f) = 1$,
$f$ has exactly one positive real root. The rule easily extends to an
arbitrary open interval $I=(a,b)$ via a suitable coordinate transformation: The mapping $x \mapsto a+(b-a)x$ maps $(0,1)$ bijectively onto $I$, that is, the roots of $f$ in $I$ exactly correspond to those of 
\begin{align}
f_I(x):=f_{[a,w(I)]}(x)=f(a+w(I)x)=f(a+(b-a)x)\label{polyfI}
\end{align} 
in $(0,1)$. Hence, the composition of $x \mapsto a+(b-a)x$ and $x \mapsto 1/(1+x)$ constitutes a bijective map from $(0,\infty)$ to $I$. It follows that the positive real roots of 
\begin{align*} 
f_{I,\operatorname{rev}}(x):= (1+x)^n f_I(\frac{1}{x+1})=(1+x)^n \cdot f(\frac{ax
+ b}{x + 1})\label{fIt}
\end{align*}
correspond bijectively to the real roots of $f$ in $I$. The factor $(1 + x)^n$ in the definition of
$f_{I,\rev}$ clears denominators and guarantees that $f_{I,\rev}$ is a polynomial. $f_{I,\operatorname{rev}}$ is computed from $f_I$ by reversing the coefficients followed by a Taylor shift by $1$. We now define $\var(f,I)$ as $\var(f_{I,\rev})$.

Based on Descartes' Rule of Sign, Vincent, Collins and Akritas introduced a bisection algorithm denoted $\textsc{Vca}$ for isolating the roots of $f$ in an interval $I_0$ (here, we assume that $I_0=(-1/2,1/2)$). 
We refer the reader to~\cite{Alesina-Galuzzi,Alesina-Galuzzi99,Alesina-Galuzzi00,BPR,collins-akritas:76,eigenwillig:thesis} for extensive treatments and
references.\vspace{0.5cm}
\hrule\vspace{0.2cm}
\noindent\textbf{\textsc{Vca}.} The algorithm requires that the real roots of $f$ in $I_0$ are simple, otherwise it diverges. In each step, a set $\mathcal{A}$
of active intervals is maintained. Initially, $\mathcal{A}$ contains $I_0$, and we stop as soon as $A$ is empty. In each iteration, some interval $I \in \mathcal{A}$ is
processed; If $\var(f,I)=0$, then $I$ contains no root of $f$ and we discard $I$. If $\var(f,I)=1$, then $I$ contains exactly one root of $f$ and
hence is an isolating interval for it. We add $I$ to a list $\mathcal{O}$ of isolating
intervals. If there is more than one sign change, we divide $I$ at its
midpoint $m(I)$ and add the subintervals to the set of active intervals. If $m(I)$ is a root of $f$, we add the trivial interval $[m(I),m(I)]$ to the list
of isolating intervals.\vspace{0.2cm}
\hrule\vspace{0.5cm} 
Correctness of the algorithm is obvious. Termination and complexity
analysis of the \textsc{Vca} algorithm rest on the
following theorem:

\begin{thm}[\cite{Obreshkoff25,Ostrowski:1950}]\label{Circle-Theorems}
Consider a polynomial $f\in\R[x]$, an interval $I=(a,b)$ and $v = \var(f,I)$. 
\begin{itemize}
\item[(i)] (One-Circle Theorem) If the open disc bounded by the circle centered at $m(I)$ and passing through the endpoints of $I$ contains no root of
$f(x)$, then $v = 0$.
\item[(ii)] (Two-Circle Theorem) If the union of the open discs bounded by the two circles centered at $m(I) \pm i
(1/(2\sqrt{3})) w(I)$ and passing through the endpoints of $I$ contains
exactly one root of $f(x)$, then $v = 1$. 
\end{itemize}\end{thm}

Proofs of the one- and two-circle theorems can be found
in~\cite{Alesina-Galuzzi,eigenwillig:thesis,Krandick-Mehlhorn:Descartes,Obreshkoff25,Obreshkoff:book,Obreshkoff:book-english,Ostrowski:1950}.
Theorem~\ref{Circle-Theorems} implies that no interval $I$ of length $\sigma_f$
or less is split. Such an interval, recall that it is open, cannot contain two real roots
and its two-circle region cannot contain any nonreal root. Thus, $\var(f,I) \le 1$ by
Theorem~\ref{Circle-Theorems}. We conclude that the depth of the recursion tree is bounded by $1/\sigma_f$. Furthermore, it holds (see~\cite[Corollary 2.27]{eigenwillig:thesis} for a simple self-contained proof):

\begin{thm}\label{subad}
Let $I$ be an interval and $I_1$ and $I_2$ be two disjoint subintervals of $I$. Then, $$\var(f,I_1) + var(f,I_2) \le \var(f,I).$$
\end{thm}

According to the above theorem, there cannot be more than $n/2$ intervals $I$ with $\var(f,I) \ge 2$ at any level of the recursion. Therefore, the size of the recursion tree $T_{\textsc{Vca}}$ is bounded by $-n\log\sigma_f$. For polynomials with integer coefficients of maximal bitsize $\tau$, it is shown that $-\log\sigma_f=O(n(\log n+\tau))$, thus, the latter bound writes as $\Otilde(n^2\tau)$. However, a more refined argumentation~\cite{eigenwillig:thesis} shows that $|T_{\textsc{Vca}}|$ is even bounded by $\Otilde(n\tau)$.

The  computation of $f_{I,\rev}$ at each node of the tree is costly. It is better to store with every interval $I = (a,b)$ the polynomial $f_I(x)= f(a + x(b - a))$. If $I$ is split at its midpoint $m(I)$ into $I_l = (a,m(I))$ and $I_r =
(m(I),b)$, the polynomials associated with the subintervals are
$f_{I_l}(x) = f_I(\frac{x}{2})$ and $f_{I_r}(x)= f_I(\frac{1 + x}{2}) =
f_{I_l}(1 + x)$.
Also, $f_{I,\rev}(x) =
(1 + x)^n f_I(\frac{1}{1+x})$. If the coefficients of $f$ are
integers (or dyadic fractions) of bitsize $\tau$, then the coefficients grow by $n$ bits in every 
bisection step. Thus, for a node $I$ of depth $h$, the bitsize $\tau_h$ of the coefficients of $f_I$ is given bounded by $\tau_h=\tau+nh$. Hence, using asymptotically fast 
Taylor shift (see~\cite{GG97,Ger04}), the
number of bit operations needed to compute $f_{I_l}$, $f_{I_{r}}$ and $f_{I,\rev}$ 
from $f_I$ is in $\Otilde(n(nh+\tau))$. Since the depth of the recursion tree is bounded by $\Otilde(n\tau)$, each $f_I$ has coefficients of bitsize $\Otilde(n^2\tau)$ and, thus, the cost at each node is in $\Otilde(n^3\tau)$. Eventually, the total cost for \textsc{Vca} is in $\Otilde(n^3\tau)\cdot \tilde{O}(n\tau)=\Otilde(n^4\tau^2)$.

\section{A Modified Descartes Method}\label{modifieddescartes}

In the REAL-RAM model, where exact
operations on real numbers are assumed to be available at unit
costs, the Descartes method can directly be used to isolate the real roots of the polynomial $f$ as defined in (\ref{polyf}). Then, for each node $I$ of the recursion tree, we have to compute the number $\var(f,I)=\var(f_{I,\rev})$ of sign variations for the polynomial $f_{I,\rev}$ and the sign of $f$ at the midpoint $m(I)$. However, for an actual implementation, these computations turn out to be hard in general because the coefficients of $f$ are arbitrary real numbers. To overcome this issue, we aim to only consider approximations of $f_I$ and $f_{I,\rev}$ instead. In Section~\ref{algorithm}, we will show that, for sufficiently good approximations of $f$, this approach is feasible. However, our approach does not directly apply to the Descartes method but to a slight modification of it.

For our modified version of the Descartes method, we aim to replace the inclusion predicate $\var(f,I)=1$ by a predicate used in the Bolzano method; see Corollary~\ref{inclusion}. Section~\ref{ttest} resumes some useful results which are adopted from our studies on the Bolzano method~\cite{sagraloff-yap:ceval:09} whereas, in Section~\ref{dcm}, our modified version is formulated. 

\subsection{The $\T_K^g(m,r)$-Test: Existence of Roots}\label{ttest}

For $g\in\C[x]$, $m\in\C$ and positive real values $K$ and $r$, we consider the test
	\begin{equation}
	\T^g_K(m,r):\quad t_K^g(m,r):=|g(m)| - K \sum_{k\ge 1}
		\left|\frac{g^{(k)}(m)}{k!} \right|r^k>0.\label{TK-test}
	\end{equation}
	
In order to simplify notation, we also write $\T_{K}^{g}(\Delta)$ or $\T_{K}^{g}(I)$ instead of $\T_{K}^{g}(m,r)$, where $\Delta=\Delta_{r}(m)$ or $I=(a,b)$ an interval with midpoint $m=m(I)$ and radius $r=r(I)$. If the polynomial $g$ is fixed and no mix-up is possible, we further omit the ''$g$'' and write $\T_{K}(m,r)$ for $\T^g_K(m,r)$ and $\T_{K}'(m,r)$ for $\T^{g'}_K(m,r)$. We mainly use $K=3/2$. Therefore, whenever the "$K$" is suppressed (i.e., we write $\T^g(m,r)$ instead of $\T_{3/2}^g(m,r)$), we consider $K=3/2$. Before presenting the main technical lemmata, we first summarize the following useful properties of $\T_K^g(m,r)$:
\begin{itemize}
\item If $\T^g_K(m,r)$ holds, then $\T^g_{K'}(m,r)$ holds for all $K'\leq K$ and all $r'\le r$.
\item For arbitrary values $m$, $r$ and $\lambda\neq 0$, the test $\T_K^g(m,r)$ is equivalent to $\T_K^{g_{[m,\lambda]}}(0,r/\lambda)$ because of $t_K^{g_{[m,\lambda]}}(0,r/\lambda)=t_K^g(m,r)$. In particular, for an interval $I=(a,b)$, the test $\T_K^{g_I}(0,r)$ is equivalent to $\T_K^{g}(a,rw(I))$, where $g_I(x)=g(a+w(I)x)$. 
\item For $\lambda\in\R^+$, $t^g_K(m,r)=t^{\lambda g}_K(m,r)\cdot\lambda^{-1}$ and, thus, $\T^{g}_K(m,r)$ is equivalent to $\T^{\lambda g}_K(m,r)$. Hence, $\T^{(g')_I}_K(m,r)$ and $\T^{(g_I)'}_K(m,r)$ are equivalent since $(g_I)'=(g(a+w(I)x))'=w(I)(g')_I$.
\end{itemize} 

The $\T_K^g(m,r)$-test serves as exclusion predicate but might also guarantee that a certain disc contains at most one root. We refer to~\cite[Theorem 3.2]{bes:bisolve:11} for a proof of the following lemma.

\begin{lem}
\label{lem:test}
Consider a disc $\Delta=\Delta_m(r)\subset\C$ and a polynomial $g\in\R[x]$:
\begin{itemize}
\item[(i)]
	If $\T_{K}(\Delta)$ holds for a $K\ge 1$, then $\bar{\Delta}$ contains no root of $g$ and
	\[
	(1-\frac{1}{K})|g(m)|<|g(z)|<(1+\frac{1}{K})|g(m)|
	\]
	for all $z$ in the closure $\bar{\Delta}$ of $\Delta$.
\item[(ii)]
	If $\T'_{3/2}(\Delta)$ holds, then
	$\bar{\Delta}$ contains at most one root of $g$.
\end{itemize}
\end{lem}

The $\T_{3/2}'(m,r)$-test now easily applies as an inclusion predicate:

\begin{cor}\label{inclusion} Let $I=(a,b)$ be an interval such that $\T_{3/2}^{g_I'}(0,r)$ holds for an $r\ge 1$. Then, $I$ contains a root $\xi$ of $g$ exactly if $g(a)\cdot g(b)< 0$. In the latter case, the disc $\Delta_{rw(I)}(a)$ is isolating for $\xi$.
\end{cor}

\begin{proof}
If $\T_{3/2}^{g_I'}(0,r)$ holds, then $\T_{3/2}^{g'}(a,rw(I))$ holds as well. It follows that the disc $\Delta_{rw(I)}(a)$ and, thus, $I$ contains no root of the derivative $g'$. Now, since $f$ is monotone on $I$, it suffices to check for a sign change of $g$ at the endpoints of $I$. Namely, there exists a root $\xi$ of $g$ in $I$ if and only if $g(a)g(b)< 0$. In case of existence, $\Delta_{rw(I)}(a)$ is isolating for $\xi$ due to Lemma~\ref{lem:test}.
\end{proof}

In order to show that the $\T'_{3/2}(m,r)$-test in combination with sign evaluation is an efficient inclusion predicate, we give lower bounds on $r$ in terms of $\sigma_g$ such that the predicate succeeds under guarantee.

\begin{lem}
\label{lem:success}
For $g$ a polynomial of degree $n$, a disc $\Delta=\Delta_r(m)\subset\C$, an interval $I=(a,b)$ and $I^+=(a-\frac{w(I)}{4n},b+\frac{w(I)}{4n})$, it holds that:
\begin{itemize}
\item[(i)] If $r\le\frac{\sigma_g}{4n^2}$, then $\T(\Delta)$ or $\T'(\Delta)$ holds.
\item[(ii)] If $\Delta$ contains a root $\xi$ of $g$ and $r\le\frac{\sigma(\xi,f)}{4n^2}$, then $\T'(\Delta)$ holds.
\item[(iii)] If $\var(g,I^+)>0$ and $\T^{g_I'}(0,2)$ fails, $\Delta_{2w(I)}(a)$ contains a root $\xi$ of $g$ with $\sigma(\xi,g)<8n^2w(I)$.
\item[(iv)] If $\var(g',I)> 0$ and $\T^{g_I}(0,1)$ fails, $\Delta_{2nw(I)}(a)$ contains a root $\xi$ of $g$ with $\sigma(\xi,g)<4n^2w(I)$.
\end{itemize}
\end{lem}  

\begin{proof}
For the proof of (i) and (ii), we refer to~\cite[Lemma 5]{s-bitstream-mcs11}. 
For (iii), suppose that $\var(g,I^+)>0$ and $\T^{g_I'}(0,2)$ does not hold. Then, according to Theorem~\ref{Circle-Theorems} (i), the disc $\Delta_{w(I^+)/2}(m(I))\subset \Delta_{2w(I)}(a)$ contains a root $\xi$ of $g$. With (ii), it follows that $2w(I)>\frac{\sigma(\xi,g)}{4n^2}$ and, thus, $\sigma(\xi,g)<8n^2w(I)$. For (iv), we first argue by contradiction that $\Delta_{2nw(I)}(a)$ contains a root $\xi$ of $g$: If $|a-x_i|\ge 2nw(I)$ for all roots $x_i$ of $g$, then 
\begin{align*}
\left|\frac{g^{(k)}(a)}{g(a)}\right|&=\left|\sum\nolimits_{i_1,\ldots,i_k}^{\prime}\frac{1}{(a-x_{i_1})\ldots(a-x_{i_k})}\right|
\leq\left(\sum\nolimits_{i=1}^n\frac{1}{|a-x_i|}\right)^k
\leq\left(\frac{1}{2w(I)}\right)^{k},
\end{align*}
where the prime means that the $i_j$'s ($j=1\ldots k$)
are chosen to be distinct. It follows that $T^g(a,w(I))$ holds because of $\sum_{k=1}^n \left|\frac{g^{(k)}(a)}{g(a)}\right| w(I)^k\le \sum_{k=1}^n 2^{-k}<1<\frac{3}{2}$. In addition, Theorem~\ref{Circle-Theorems} guarantees the 
existence of a root $\xi'\in\Delta_{w(I)/2}(m(I))$ of $g'$. Hence, we have 
$|\xi-\xi'|<2nw(I)+r(I)<4nw(I)$ which implies $\sigma(\xi,g)<4n^2w(I)$ due to 
the fact~\cite{Eigenwillig2007a,yap-fundamental} that there exists no root of 
the derivative $g'$ in $\Delta_{\sigma(\xi,g)/n}(\xi)$. 
\end{proof}\vspace{0.25cm}

\subsection{$\dcm$: A Modified Descartes Algorithm}\label{dcm}
 
We introduce our modified Descartes method \textsc{Dcm} (short for ``Descartes modified'') to isolate the real roots of a polynomial $f$ as defined in (\ref{polyf}). We formulate the algorithm in the REAL-RAM model, thus, it still does not directly apply to bitstream polynomials. However, in Section~\ref{dcmL}, we will present a corresponding version $\dcm^\rho$ of $\textsc{Dcm}$ which resolves this issue; see also Appendix, Algorithm~\ref{alg:dcm} for pseudo-code of $\dcm$.\vspace{0.5cm}
\hrule\vspace{0.2cm}
\noindent\textbf{\textsc{Dcm}.} $\dcm$ maintains a list $\mathcal{A}$ of active nodes and a list $\mathcal{O}$ of isolating intervals, where we initially set $\mathcal{O}=\emptyset$ and $\mathcal{A}:=\{(I_0,f_{I_0})\}$ with $I_0:=(-\frac{1}{2},\frac{1}{2})$. For each active node $(I,f_{I})\in\mathcal{A}$, we proceed as follows. We remove $(I,f_{I})$ from $\mathcal{A}$. Then, we compute the number $v_{I^+}:=\var(f,I^+)=\var(f_{I^{+},\rev})$ of sign variations for $f$ on the extended interval $I^+$.
We remark that $f_{I^{+}}(x)=f_{I}(-\frac{1}{4n}+(1+\frac{1}{2n})x)$ and $f_{I^{+},\rev}(x)=(1+x)^{n}f_{I^{+}}(\frac{1}{1+x})$.
If $v_{I^+}=0$, we do nothing. If $v_{I^+}\ge 1$, we consider the test $\T^{f_I'}(0,2)$ which is equivalent to $\T^{f'}(a,2w(I))$. If it fails, then $I$ is subdivided into $I_{l}=(a,m(I))$ and $I_{r}=(m(I),b)$ and we add $(I_{l},f_{I_{l}})=(I_{l},f_{I}(\frac{x}{2}))$ and $(I_{r},f_{I_{r}})=(I_{r},f_{I_{l}}(x+1))$ to $\mathcal{A}$. Otherwise, we evaluate the sign $s$ of $f(a^+)\cdot f(b^+)=f_{I^{+}}(0)\cdot f_{I^{+}}(1)$. If $s<0$ and $I^+$ is disjoint from any other interval in $\mathcal{O}$, we add $I^+$ to $\mathcal{O}$.
If $s\ge 0$ or $I$ intersects an interval in $\mathcal{O}$, we do nothing. The algorithm stops when $\mathcal{A}$ becomes empty.\vspace{0.2cm}
\hrule\vspace{0.5cm}

\begin{thm}\label{dcmltermination}
For the polynomial $f$ as defined in (\ref{polyf}), $\dcm$ terminates and returns a list $\mathcal{O}=\{I_1,\ldots,I_m\}$ of disjoint isolating intervals for \emph{all} real roots of $f$.
\end{thm}

\begin{proof}
If the width $w(I)$ of an interval $I=(a,b)$ is smaller or equal to $\frac{\sigma_f}{8n^2}$, then, according to Theorem~\ref{Circle-Theorems}, $\var(f,I^+)=0$ or $\T^{f_I'}(0,2)$ holds. Thus, $I$ is not further subdivided. This shows termination of $\dcm$. From our construction and Corollary~\ref{inclusion}, each interval in $\mathcal{O}$ is isolating for a real root of $f$ and all intervals in $\mathcal{O}$ are pairwise disjoint. It remains to show that, for each real root $\xi$ of $f$, there exists a corresponding isolating interval in $\mathcal{O}$. Since all roots of $f$ have absolute value bounded by $1/2$, there must be a terminal interval $I=(a,b)$ whose closure $\bar{I}$ contains $\xi$. Since $v_{I^+}>0$, $I$ cannot be discarded in the first step of $\dcm$. Hence, $\T^{f_I'}(0,2)$ holds and, thus, $f$ is monotone on $I^+$. Since $I^+$ contains the root $\xi$, we have $f(a^+)f(b^+)<0$. It follows that either $I^+$ is added to the list of isolating intervals or $I^+$ intersects an interval $J^+=(c^+,d^+)\in\mathcal{O}$ which has been added to $\mathcal{O}$ before. Let $J=(c,d)$ be the corresponding smaller interval for $J^+$. Since the $\frac{w(I)}{4n}$-neighborhood of $I$ intersects the $\frac{w(J)}{4n}$-neighborhood of $J$, the preceding Lemma~\ref{lem:intersection} shows that one of the discs $\Delta_{2w(I)}(a)$ or $\Delta_{2w(J)}(c)$ contains both intervals $I^+$ and $J^+$. Since both $\T^{f_I'}(0,2)$ and $\T^{f_J'}(0,2)$ hold, each of the latter two discs contains at most one root due to Corollary~\ref{inclusion}. It follows that $J^+\in\mathcal{O}$ already isolates $\xi$.
\end{proof}
\ignore{
\begin{figure}[t]
\begin{center}
\vspace{-7.7cm}\includegraphics[width=12cm]{twocircles}\end{center}
\vspace{-3.75cm}\caption{\label{fig:twocircles} \footnotesize{Wlog., we can assume that $w(J)\ge w(I)$. $w(I)$, $w(J)$ and the distance $\delta$ between $I$ and $J$ differ by a power of $2$. For $\delta=0$, the disc $\Delta:=\Delta_{2w(J)}(c)$ certainly contains $\tilde{I}$ and $\tilde{J}$. If $\delta\neq 0$, then $w(J)\ge 2w(I)$ and $w(J)\ge 4\delta$, hence $\tilde{I}$, $\tilde{J}\subset\Delta$.}}
\end{figure}
}
\begin{lem}\label{lem:intersection}\footnote{Lemma~\ref{lem:intersection} proves a slightly stronger result than necessary for the proof of Theorem~\ref{dcmltermination}. The stronger result applies in the proof of Theorem~\ref{thm:success} in Section~\ref{sec:knownvalues}.}
Let $I=(a,b)$ and $J=(c,d)$ be two intervals (not necessarily of equal length) of the form $\left(-\frac{1}{2}+i 2^{-h},-\frac{1}{2}+(i+1)2^{-h}\right),$ where $h\in\N$ and $i\in\{0,\ldots,2^h-1\}$. If the $\frac{w(I)}{2n}$-neighborhood $U_{w(I)/2n}(I)$ of $I$ intersects the $\frac{w(J)}{2n}$-neighborhood $U_{w(J)/2n}(J)$ of $J$, then one of the discs $\Delta_{2w(I)}(a)$ or $\Delta_{2w(J)}(c)$ contains the intervals $(a-w(I),b+w(I))$ and $(c-w(J),d+w(J))$.   
\end{lem}

\begin{proof}
W.l.o.g., we can assume that $w(J)\ge w(I)$ and, thus, $w(J)=2^{l}w(I)$ with an $l\in\N_0$. Let $\delta$ denote the distance between $I$ and $J$. If $\delta=0$, then $\Delta_{2w(J)}(c)$ contains $(a-w(I),b+w(I))$ and $(c-w(J),d+w(J))$. If $\delta\neq 0$, then $\delta=2^kw(I)$ with a $k\in\N_0$. Since $U_{w(I)/2n}(I)\cap U_{w(J)/2n}(J)\neq \emptyset$, we must have $\frac{w(J)}{2n}>\frac{\delta}{2}$. In particular, we have $\frac{w(J)}{4}>\frac{\delta}{2}=2^{k-1}w(I)$. Since $w(I)$ and $w(J)$ differ by a power of $2$, it follows that $w(J)\ge 2^{k+2}w(I)=4\delta$ and, thus, $2w(J)=w(J)+\frac{w(J)}{2}+\frac{w(J)}{2}\ge w(J)+ 2w(I)+2\delta$. From the latter inequality our claim follows.
\end{proof}

\begin{thm}\label{thm:complexitydcm}
For a polynomial $f$ as in (\ref{polyf}), $\dcm$ induces a subdivision tree $T_{\dcm}$ of 
$$\text{\emph{height} } h(T_{\dcm})=O(\log n-\log\sigma_f)\text{ and \emph{size} } |T_{\dcm}|=O(\Sigma_f+n\log n).$$
\end{thm}

\begin{proof}
The result on the height of $T_{\dcm}$ follows directly from the proof of Theorem~\ref{dcmltermination}. Namely, we have shown that $\dcm$ never subdivides an interval of width less than or equal to $\frac{\sigma_f}{8n^2}$.
For the bound on $|T_{\dcm}|$, we use a similar argument as in~\cite{eigenwillig-sharma-yap:descartes:06} and~\cite{ms-detbitdesc-09}. Namely, for a root $\xi$ of $f$ and a certain $h\in\mathbb{N}_0$ we say that $I=(-\frac{1}{2}+i2^{-h},-\frac{1}{2}+(i+1)2^{-h})$, $i=\{0,\ldots,2^h-1\}$, is \emph{a canonical interval} for $\xi$ if the real part of $\xi$ is contained in $[-\frac{1}{2}+i2^{-h},-\frac{1}{2}+(i+1)2^{-h})$ and $\sigma(\xi,f)<8n^2 2^{-h}=8n^2w(I)$. We denote $T_c$ the \emph{canonical tree} which consists of all canonical intervals. We remark that, for a canonical interval $I$, the parent interval of $I$ is canonical as well. The following considerations will show that $|T_{\dcm}|=O(|T_{c}|)$ and $|T_{c}|=O(\Sigma_f+n\log n)$. For the size of the canonical tree, consider a leaf $I\in T_{c}$ and let $\xi_I$ be a root of $f$ corresponding to this leaf. If there are several, then $\xi_I$ is the root with minimal separation. Then, $\sigma(\xi_I,f)<8n^2 2^{-h}$ and, thus, $h\le 2\log n+4-\log\sigma(\xi_I,f)$. Since each root of $f$ is associated with at most one leaf of the canonical tree, we conclude $|T_c|=O(n\log n+\Sigma_f).$
It remains to show that $|T_{\dcm}|=O(|T_c|)$. Consider the following mapping of internal nodes (intervals) of $T_{\dcm}$ to canonical nodes (intervals) in $T_c$: Let $I$ be a non-terminal interval of width $w(I)=2^{-h}$. Then, $\var(f,I^+)> 0$ and $T^{f_I'}(0,2)$ does not hold. According Lemma~\ref{lem:success} (iii), the disc $\Delta_{2w(I)}(a)$ contains a root $\xi$ of $f$ with $\sigma(\xi,f)<8n^2 w(I)=8n^2 2^{-h}$. Hence, one of the four intervals $I_1=(a-2w(I),a-w(I))$, $I_2=(a-w(I),a)$, $I$ or $I_2=(b,b+(b-a))$ is canonical for $\xi$. We map $I$ to the corresponding interval. This defines a mapping from the internal nodes of $T_{\dcm}$ to the nodes of the canonical tree $T_c$. Furthermore, each node in the canonical tree has at most four preimages in $T_{\dcm}$ and, thus, the number of internal nodes of $T_{\dcm}$ is bounded by $O(n\log n+\Sigma_f)$.
Since $T_{\dcm}$ is a binary tree, the bound on the number of
internal nodes applies to the whole tree as well. 
\end{proof}

\section{Algorithm}\label{algorithm}

We first outline our algorithm $\R\textsc{Isolate}$ to isolate the roots of $f$. $\R\textsc{Isolate}$ decomposes into two subroutines $\dcm^\rho$ and $\textsc{Certify}^\rho$, where $\rho$ indicates the actual working precision. $\dcm^\rho$ is essentially identical to $\dcm$ with the main difference that, at each node $I=(a,b)$ of the recursion tree, we only consider approximations $\tilde{f}_I(x)$ of $f_I(x)=f(a+w(I)x)$ to a certain number $\rho_I$ of bits after the binary point, where $\rho+2\log w(I)\le\rho_I\le \rho$. We remark that we proceed $I$ in a way such that it is terminal for $\dcm^\rho$ if it is terminal for the exact counterpart $\dcm$. This ensures that, for any $\rho$, $\dcm^\rho$ induces a subtree $T_{\dcm^{\rho}}$ of $T_{\dcm}$ and, thus, $|T_{\dcm^{\rho}}|=O(\Sigma_f+n\log n)$ due to Theorem~\ref{thm:complexitydcm}. We further show that, for a precision $\rho\ge\rho_f^{\max}=O(\Sigma_f+n)$, $\dcm^\rho$ returns isolating intervals for \emph{all} real roots of $f$; see Theorem~\ref{thm:success} for the definition of $\rho_f^{\max}$ and further details. However, for smaller $\rho$, $\dcm^\rho$ may return isolating intervals only for some roots but without any information whether all real roots are captured or not. In order to overcome such an undesirable situation, we consider an additional subdivision method $\textsc{Certify}^\rho$ similar to $\dcm^\rho$ which aims to certify that all roots are captured. We further show that $\textsc{Certify}^\rho$ also induces a recursion tree of size $O(\Sigma_f+n\log n)$ and succeeds if $\rho\ge\rho_f^{\max}$. If, for a given precision $\rho$, our algorithm fails to isolate all roots of $f$, we double $\rho$ and restart.

\subsection{$\dcm^\rho$: An Approximate Version of $\dcm$}\label{dcmL}

We present our first subroutine $\dcm^\rho$. Comments to support the approach are in \emph{italic} and marked by a "//" at the beginning.\vspace{0.5cm}

\ignore{
In this section, we formulate a version of $\dcm$ which only considers approximations $\tilde{f}_I$ of the polynomials $f_I$ in each recursion step. The algorithm is driven by an initial precision $L\in\N$ explaining our denotation $\dcm^L$ (cf. Algorithm~\ref{alg:dcml} in the Appendix for pseudo-code). The algorithm is formulated in a way such that, for arbitrary precision $L$, $\dcm^L$ induces a subtree of the recursion tree of the exact counterpart $\dcm$ when applied to $f$ (cf. Theorem~\ref{complexitydcml} for a proof):\vspace{0.5cm}}
\hrule\vspace{0.2cm}
\noindent\textbf{$\dcm^{\rho}$.}
Let $I_0=(-\frac{1}{2},\frac{1}{2})$ be the starting interval which, by construction of $f$, contains all real roots of $f$. In a first step, we choose a $(\rho+n+1)$-binary approximation $\tilde{f}$ of $f$ and evaluate $\tilde{f}(-\frac{1}{2}+x)$. Then, the resulting polynomial is approximated by a $(\rho+1)$-binary approximation $\tilde{f}_{I_0}\in [\tilde{f}(-\frac{1}{2}+x)]_{2^{-\rho-1}}$ and, according to Lemma~\ref{lem:taylorshift}, we have $\tilde{f}_{I_0}\in [f_{I_0}]_{2^{-\rho}}$.\\ 

$\dcm^\rho$ maintains a list $\mathcal{A}$ of active nodes $(I,\tilde{f}_I,\rho_I)$, where $I=(a,b)\subset I_0$ is an interval, $\tilde{f}_I$ approximates $f_I$ to $\rho_I$ bits after the binary point and $\rho+2\log w(I)\le\rho_I\le \rho$. $\dcm^\rho$ eventually returns a list $\mathcal{O}$ of tuples $(J,s_{J,l},s_{J,r},B_J)$, where $J=(c,d)$ is an isolating interval for a root of $f$, $s_{J,l}=\sgn f(c)$, $s_{J,r}=\sgn f(d)$ and $0<B_J\le\min(|f(c)|,|f(d)|)$. We initially start with $\mathcal{A}:=\{(I_0,\tilde{f}_{I_0},L)\}$ and $\mathcal{O}:=\emptyset$. For each active node, we proceed as follows:\vspace{0.25cm}
\begin{enumerate} 
\item Remove $(I,\tilde{f}_I,\rho_I)$ from $\mathcal{A}$.  
\item Compute the polynomials \begin{align}
\tilde{f}_{I^+}(x)=\tilde{f}_I(-\frac{1}{4n}+(1+\frac{1}{2n})x) \text{ \hspace{0.1cm} and  \hspace{0.1cm}} 
\tilde{h}(x)=\sum_{i=0}^n 
\tilde{h}_i x^i:=(1+x)^n \tilde{f}_{I^+}(\frac{1}{1+x}).\label{polyh}
\end{align}
\item If $\tilde{h}_i>-2^{n+2-\rho_I}$ for all $i$ or $\tilde{h}_i<2^{n+2-\rho_I}$ for all $i$, do nothing (i.e., $I$ is dicarded).\\

\noindent$\text{\large{//}}$\textit{ A simple computation (see the subsequent Lemma~\ref{lem:dcml} (i)) shows that 
$\tilde{h}$  approximates $f_{I^+,\rev}(x)=(x+1)^n f_{I^+}(\frac{1}{1+x})$ to $\rho_I-n-2$ bits after the binary point. Thus, if $\var(f,I^+)=0$, all coefficients of $\tilde{h}$ are either smaller than $2^{n+2-\rho_I}$ or larger than $-2^{n+2-\rho_I}$. Since we want to induce a subtree of the recursion tree $T_{\dcm}$ induced by $f$, we discard $I$ if all coefficients of $\tilde{h}$ are larger than $-2^{n+2-\rho_I}$ (or smaller than $2^{n+2-\rho_I}$).}\\ 

\item If there exist $\tilde{h}_i$ and $\tilde{h}_j$ with $\tilde{h}_i\le -2^{n+2-\rho_I}$ and $\tilde{h}_j\ge 2^{n+2-\rho_I}$, consider the test $\T^{(\tilde{f}_I)'}(0,2)$, that is, evaluate $t_{3/2}^{(\tilde{f_I})'}(0,2)$.\\
\noindent$\text{\large{//}}$\textit{ Due to Lemma~\ref{lem:dcml} (i), we have $|t_{3/2}^{(f_I)'}(0,2)-t_{3/2}^{(\tilde{f_I})'}(0,2)|<n2^{n+1-\rho_I}$. Hence, if $\T^{(f_I)'}(0,2)$ holds, then $t_{3/2}^{(\tilde{f_I})'}(0,2)>-n2^{n+1-\rho_I}$. Thus, we proceed as follows:}\\
\begin{itemize} 
\item[(a)] If $t_{3/2}^{(\tilde{f_I})'}(0,2)>-n2^{n+1-\rho_I}$, consider the polynomial 
\begin{align}
\hat{f_I}(x):=\tilde{f_I}(x)+n2^{n+1-\rho_I}\cdot x\label{fhat},\vspace{0.25cm}
\end{align}
$\text{\large{//}}$\textit{ Then, $\T^{(\hat{f_I})'}(0,2)$ holds and, in particular, $\hat{f_I}$ is monotone on $(-2,2)$.}\\

evaluate 
\begin{align}
\lambda^-&:=\hat{f_I}(-\frac{1}{4n})=\tilde{f}_{I^+}(0)-2^{n-1-\rho_I}\label{lambda-}\\
\lambda^+&:=\hat{f_I}(1+\frac{1}{4n})=\tilde{f}_{I^+}(1)+(4n+1)2^{n-1-\rho_I}\label{lambda+}\\
\lambda&:=\hat{f_I}(-\frac{1}{n})=\tilde{f_{I}}(-\frac{1}{n})-2^{n+1-\rho_I},\label{lambda2}
\end{align}
and check whether the following conditions are fulfilled:
\begin{align}
\tilde{I}=(\tilde{a},\tilde{b})=(a-\frac{w(I)}{2n},b+\frac{w(I)}{2n})\text{ intersects no }J\text{ for any }(J,s_{J,l},s_{J,r},B_J)\in\mathcal{O},\label{cond1}
\end{align}
\vspace{-0.65cm}
\begin{align}
\lambda^-\cdot\lambda^+&< 0,\label{ineq1}\\
\min(|\lambda^-|,|\lambda^+|)&>2^{n+3-\rho_I}n, \text{ and} \label{ineq3}\\
|\lambda|&>2^{\deg \hat{f_I}+n+7-\rho_I}n^2.\label{ineq4}
\end{align}
If any of the conditions (\ref{cond1})-(\ref{ineq4}) fails, do nothing. If all conditions are fulfilled, then add $(\tilde{I},\sgn\lambda^-,\sgn\lambda^+,\min(|\lambda^-|,|\lambda^+|)-2^{n+3-\rho_I}n)$ to $\mathcal{O}$.\\

\noindent$\text{\large{//}}$\textit{ If (\ref{ineq1})-(\ref{ineq4}) hold, then $\tilde{I}$ is isolating for a root $\xi$ of $f$; see Lemma~\ref{lem:dcml} (iii). Furthermore, since $\hat{f}_I$ is monotone on $(-2,2)$, we have $|\hat{f}_I(-\frac{1}{2n})|>|\lambda^-|$ and $|\hat{f}_I(1+\frac{1}{2n})|>|\lambda^+|$. Then, from inequality (\ref{ineq3}) and Lemma~\ref{lem:dcml} (ii), it follows that $\sgn f(\tilde{a})=\sgn(\lambda^-)$, $\sgn f(\tilde{b})=\sgn(\lambda^+)$ and $\min(|f(\tilde{a})|,|f(\tilde{b})|)>\min(|\lambda^-|,|\lambda^+|)-2^{n+3-\rho_I}n$.}\\

\item[(b)] If $t_{3/2}^{(\tilde{f_I})'}(0,2)\le -n2^{n+1-\rho_I}$, subdivide $I$ into $I_l:=(a,m_I)$ and $I_r:=(m_I,b)$. Compute a $\rho_I$-binary approximation $\tilde{f}_{I_l}$ of $\tilde{f}_{I}(\frac{x}{2})$ and a $(\rho_I-1)$-binary approximation $\tilde{f}_{I_r}$ of $\tilde{f_I}(\frac{x+1}{2})$, and add $(I_l,\tilde{f}_{I_l},\rho_I-1)$ and
$(I_r,\tilde{f}_{I_r},\rho_I-2)$ to $\mathcal{A}$. If $\rho_I<2$, return ``insufficient precision''.\\

\noindent$\text{\large{//}}$\textit{ Due to Lemma~\ref{lem:taylorshift}, we have $\tilde{f}_{I_l}\in [f_{I_l}]_{2^{-\rho_I-1}}$ and $\tilde{f}_{I_r}\in [f_{I_r}]_{2^{-\rho_I-2}}$. Hence, by induction, it follows that $\rho+2\log w(I)\le\rho_I\le \rho$ for all active nodes.}\\ 
\end{itemize}
\end{enumerate}
$\dcm^\rho$ stops when $\mathcal{A}$ becomes empty. It may either return "insufficient precision" (in Step 4 (b)) or a list $\mathcal{O}$ of isolating intervals $\tilde{I}$ for some of the roots of $f$ together with the signs of $f$ and a lower bound on $|f|$ at the endpoints of $\tilde{I}$.\vspace{0.2cm}
\hrule\vspace{0.5cm}

\begin{lem}\label{lem:dcml}
Let $f$ be a polynomial as in (\ref{polyf}), $I=(a,b)$ an interval considered by $\dcm^\rho$ and $\tilde{h}$ the polynomial as defined in (\ref{polyh}). Then,
\begin{itemize}
\item[(i)] $\tilde{h}(x)\in [f_{I^+,\rev}]_{2^{n+2-\rho_I}}$ \text { and }$|t_{3/2}^{(f_I)'}(0,2)-t_{3/2}^{(\tilde{f}_I)'}(0,2)|<n\cdot 2^{n+1-\rho_I}.$
\item[(ii)] For an arbitrary real value $t$ with $|t|\le 1+\frac{1}{n}$, it holds that $|f(a+t\cdot w(I))-\hat{f_I}(t)|<2^{n+3-\rho_I}n$, with $\hat{f}_I$ as defined in (\ref{fhat}). In particular,
\begin{align*}
|f(a^+)-\lambda^{-}|,\text{ }|f(a-\frac{w(I)}{n})-\lambda|,\text{ }|f(b^+)-\lambda^{+}|<2^{n+3-\rho_I}n,
\end{align*}
with $\lambda^{-},$ $\lambda^{+}$ and $\lambda$ as defined in (\ref{lambda-})-(\ref{lambda2}).
\item[(iii)]
Suppose that $t_{3/2}^{(\tilde{f_I})'}(0,2)>-n2^{n+1-\rho_I}$ and the inequalities (\ref{ineq1})-(\ref{ineq4}) hold. Then, $I^+$ contains a real root $\xi$ of $f$ and the $\frac{w(I)}{n}$-neighborhood of $I$ is isolating for $\xi$.
\item[(iv)] For any tuple $(J,s_{J,l},s_{J,r},B_J)\in\mathcal{O}$, the endpoints of $J$ are located outside the union of the discs $$\Delta_i:=\Delta_{\sigma(z_i,f)/(64n^3)}(z_i),\text{ where }i=1,\ldots,n.$$ 
\end{itemize}
\end{lem}

\begin{proof}
Since $\tilde{f_I}\in [f_I]_{2^{-\rho_I}}$, we have $\tilde{f}_{I^+}\in [f_{I^+}]_{2^{-\rho_I+2}}$ due to Lemma~\ref{lem:taylorshift} (ii). Reversing the coefficients and replacing $x$ by $x+1$ increases the error by a factor of at most $2^n$ (see Lemma~\ref{lem:taylorshift} (iii)), thus $\tilde{h}\in [f_{I^+,\rev}]_{2^{-\rho_I+2+n}}$. For the second part of (i), consider the following simple computation:
$$|t_{3/2}^{(f_I)'}(0,2)-t_{3/2}^{(\tilde{f_I})'}(0,2)|\le \frac{3}{2}\cdot n\cdot 2^{-\rho_I}\sum_{i=0}^{n-1}2^i=\frac{3}{2}\cdot n \cdot 2^{-\rho_I}(2^n-1)<n2^{n+1-\rho_{I}},$$ 
where the first inequality uses $(\tilde{f_I})'\in [(f_I)']_{n\cdot 2^{-\rho_I}}$.

\begin{figure}[t]
\begin{center}
\vspace{-8.45cm}\includegraphics[width=12cm]{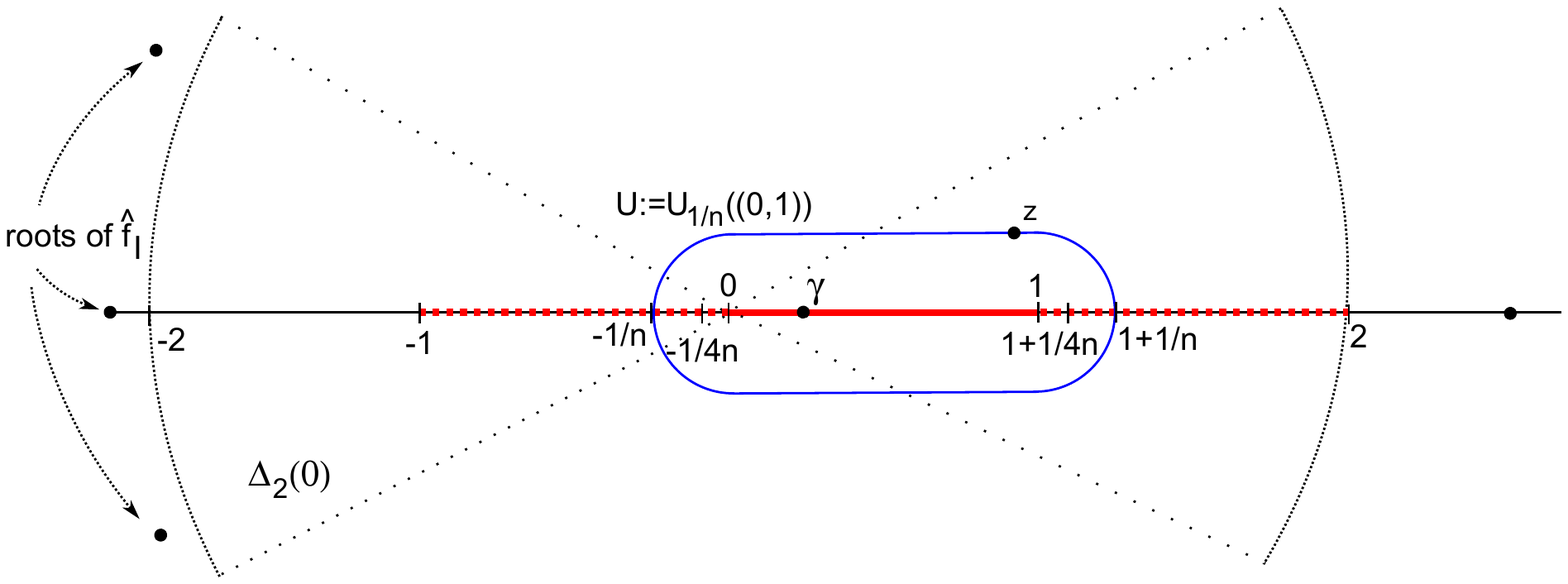}\end{center}
\vspace{-4.00cm}\caption{\label{fig:inklusion} \footnotesize{If $\lambda^-\cdot\lambda^+=\hat{f_I}(-\frac{1}{4n})\cdot\hat{f_I}(1+\frac{1}{4n})<0$, then there exists a root $\gamma\in(-\frac{1}{4n},1+\frac{1}{4n})$ of $\hat{f_I}$. Furthermore, $\Delta_{2}(0)$ contains no further root of $\hat{f}_I$. A computation shows that $|\hat{f_I}(z)|>|f_I(z)-\hat{f_I}(z)|$ for all $z$ on the boundary of the $\frac{1}{n}$-neighborhood $U$ of $(0,1)$ if the inequality (\ref{ineq4}) holds. Then, due to Rouch\'{e}'s Theorem, $U$ isolates a root of $f_I$.}}
\end{figure}

For (ii), we have
\begin{align*}
|f(a+tw(I))-\hat{f}_I(t)|&= |f_I(t)-\hat{f_I}(t)|\le|f_I(t)-\tilde{f}_{I}(t)|+|t|\cdot 2^{n+1-\rho_I}n\\
&\le 2^{-\rho_I}\sum_{i=0}^n |t|^i+(1+\frac{1}{n})2^{n+1-\rho_I}n\le n2^{-\rho_I}(1+\frac{1}{n})^{n+1}+n2^{n+2-\rho_I}<n2^{n+3-\rho_I}.
\end{align*}
\ignore{
and, in complete analogy, $|f(a-\frac{w(I)}{n})-\lambda|<2^{n+2-\rho_I}$.
For the right endpoint $b^{+}$, we have
\begin{align*}
\left|f(b^+)-\lambda^{+}\right|&=\left|f_I(1+\frac{1}{4n})-\hat{f_I}(1+\frac{1}{4n})\right|\le n2^{n+2-\rho_I}+(n+1)2^{-\rho_I}(1+\frac{1}{4n})^n<n2^{n+3-\rho_{I}}.
\end{align*}
}
Now, if the inequalities (\ref{ineq1}) and (\ref{ineq3}) hold, then $\sgn f(a^+)=\sgn(\lambda^-)$, $\sgn f(b^+)=\sgn(\lambda^+)$ and $f(a^+)\cdot f(b^+)<0$, hence, $f$ has a real root in $I^+$. We next show that (\ref{ineq4}) implies the uniqueness of this root. From $t_{3/2}^{(\tilde{f_{I}})'}(0,2)>-n2^{n+1-\rho_{I}}$, it follows that $\T_{3/2}^{(\hat{f_I})'}(0,2)$ succeeds and, thus, $\Delta_2(0)$ contains at most one root of $\hat{f_I}$. Since $\lambda^{-}=\hat{f_I}(-\frac{1}{4n})$ and $\lambda^{+}=\hat{f_I}(1+\frac{1}{4n})$ have different signs, the interval $(-\frac{1}{4n},1+\frac{1}{4n})$ contains a root $\gamma$ of $\hat{f_I}$. We consider the $\frac{1}{n}$-neighborhood $U\subset\C$ of $(0,1)$ and an arbitrary point $z$ on its boundary; see Figure~\ref{fig:inklusion}. It holds that $|-\frac{1}{n}-\gamma|/|z-\gamma|<(1+\frac{3}{4n})/(\frac{1}{4n})=4n+3<8n$ and, for any root $\tilde{\gamma}\neq \gamma$ of $\hat{f_I}$, we have $$\frac{|-\frac{1}{n}-\tilde{\gamma}|}{|z-\tilde{\gamma}|}\le \frac{|-\frac{1}{n}-z|+|z-\tilde{\gamma}|}{|z-\tilde{\gamma}|}\le 1+\frac{1+\frac{2}{n}}{1-\frac{1}{n}}=2\frac{1+\frac{1}{2n}}{1-\frac{1}{n}}.$$ Hence, it follows that
\begin{align*}
\left|\frac{\lambda}{\hat{f_I}(z)}\right|&=\left|\frac{\hat{f_{I}}(-\frac{1}{n})}{\hat{f_I}(z)}\right|=\frac{|-\frac{1}{n}-\gamma|}{|z-\gamma|}\prod_{\tilde{\gamma}\neq\gamma:\hat{f_{I}}(\tilde{\gamma})=0} \frac{|-\frac{1}{n}-\tilde{\gamma}|}{|z-\tilde{\gamma}|}<n2^{\deg\hat{f_I}+2}(1+\frac{1}{2n})^{\deg \hat{f_I}-1}(1-\frac{1}{n})^{-\deg \hat{f_I}+1}\\
&<n2^{\deg\hat{f_I}+2}\cdot\sqrt{2.72}\cdot 2.72<n2^{\deg\hat{f_I}+4}
\end{align*}
and, thus, $|\hat{f_I}(z)|>|\lambda|\cdot 2^{-\deg \hat{f_I}-4}n^{-1}$. Since $|z|\le 1+\frac{1}{n}$, we have $|f_I(z)-\hat{f_I}(z)|<n2^{n+3-\rho_I}$ according to (ii). 
Then, from Rouch\'{e}'s Theorem, it follows that $f_I$ has exactly one root within $U$ if (\ref{ineq4}) holds. This shows (iii). It remains to prove (iv): Let $J=\tilde{I}=(\tilde{a},\tilde{b})$ and $I=(a,b)$ the corresponding smaller interval. From our construction and (iii), $I^+$ contains a root 
$\xi=z_{i_0}$ of $f$ and the $\frac{w(I)}{n}$-neighborhood of $I$ is isolating for this root, thus, $|\tilde{a}-z_i|>\frac{w(I)}{4n}$ for all $i$. If there exists an $i\neq i_0$ with $\tilde{a}\in\Delta_i$, then
$w(I)<4n|\tilde{a}-z_i|<\sigma(z_i,f)/(16n^2)$. Thus, we obtain 
\begin{align*}
|\xi-z_i|&\le|\xi-\tilde{a}|+|\tilde{a}-z_i|<(1+\frac{1}{n})w(I)+\frac{\sigma(z_i,f)}{64n^3}<(1+\frac{1}{n})\frac{\sigma(z_i,f)}{16n^2}+\frac{\sigma(z_i,f)}{64n^3}<\sigma(z_i,f),
\end{align*} 
a contradiction. It remains to show that $\tilde{a}\notin 
\Delta_{i_0}$. If $\tilde{a}\in\Delta_{i_0}$, then  $w(I)<\frac{\sigma(\xi,f)}{16n^2}$. According to Lemma~\ref{lem:success} (ii), $T_{3/2}^{(f_J)'}(0,2)$ already holds for a parent node $J$ of $I$ and, thus, $t_{3/2}^{(\tilde{f_J})'}(0,2)>n2^{n+1-\rho_J}$ because of (i). This contradicts the fact that $J$ is not terminal. In completely analogous manner, one shows that $\tilde{b}$ is also not contained in any $\Delta_i$. This proves (iv). 
\end{proof}

We close this section with a result on the size of the recursion tree induced by $\dcm^{\rho}$ and the bit complexity of $\dcm^{\rho}$:

\begin{thm}\label{complexitydcml}
Let $f$ be a polynomial as in (\ref{polyf}) and $\rho\in\mathbb{N}$ an arbitrary positive integer. Then, the recursion tree $T_{\dcm^\rho}$ induced by $\dcm^{\rho}$ is a subtree of the tree $T_{\dcm}$ induced by $\dcm$, thus,
\begin{align*}
|T_{\dcm^\rho}|\le |T_{\dcm}|=O(\Sigma_f+n\log n).
\end{align*}
Furthermore, $\dcm^{\rho}$ demands for a number of bit operations bounded by
$$\Otilde(n(\Sigma_f+\log n)(n\Gamma+\tau+\rho-\log\sigma_f)).$$ 
\end{thm}

\begin{proof} 
For the first claim, we remark that $\dcm^{\rho}$ never splits an interval $I$ which is not split by $\dcm$ when applied to the exact polynomial $f$. Namely, if $I$ is terminal for $\dcm$, then either $t_{3/2}^{(f_I)'}(0,2)>0$ or $\var(f,I^+)=\var(f_{I^+,\rev})=0$. In the first case, we must have $t_{3/2}^{(\tilde{f_I})'}(0,2)>-n2^{n+1-\rho_I}$ whereas, in the second case, all coefficients $\tilde{h}_i$ of $\tilde{h}(x)=(1+x)^n\tilde{f}_{I^+}(\frac{1}{1+x})$ are either larger than $-2^{n+2-\rho_I}$ or smaller than $2^{n+2-\rho_I}$; see Lemma~\ref{lem:dcml} (i). Thus, $I$ is terminal for $\dcm^\rho$ as well. The result on the size of $T_{\dcm^\rho}$ then follows directly from Theorem~\ref{thm:complexitydcm}. 

For the bit complexity, we first consider the cost in each iteration: For an active node $(I,\tilde{f_I},\rho_I)\in\mathcal{A}$, $I=(a,b)$, the polynomial $\tilde{f_I}$ approximates $f_I$ to $\rho_I\le\rho$ bits after the binary point. The absolute value of each coefficient of $f_I$ is bounded by $2^{n+\tau}(2\Gamma)^n$ because the shift operation $x\mapsto a+(b-a)x$ does not increase the coefficients of $f$ by a factor of more than $2^n$ and the absolute value of the coefficients of $f$ is bounded by $2^{\tau+n(\Gamma+1)}$; see Section~\ref{sec:scaling}. It follows that the bitsize of the coefficients of $\tilde{f}_I$ is bounded by $n(\Gamma+1)+\tau+\rho$. Hence, the cost for computing $\tilde{h}(x)$, $\tilde{f_{I_l}}$ and $\tilde{f_{I_r}}(x)$ is bounded by $\Otilde(n(n\Gamma+\tau+\rho))$. Namely, the latter constitutes a bound on the cost for a fast asymptotic Taylor shift by an $O(\log n)$-bit number. The cost for evaluating $t_{3/2}^{(f_I)'}(0,2)$, $\lambda^-$, $\lambda^+$ and $\lambda$ matches the same bound because all these computations are evaluations of a polynomial of bitsize $O(n\Gamma+\tau+\rho)$ at an $O(\log n)$-bit number. We further remark that, in each iteration, $\mathcal{O}$ contains disjoint isolating intervals $J$ for some of the real roots of $f$ and, thus, $|\mathcal{O}|\le n$. Hence, the endpoints of the interval $J$ have to be compared with those of at most $n$ intervals stored in $\mathcal{O}$. Since $\dcm^\rho$ does not produce any interval of size less than $\frac{\sigma_f}{8n^2}$, these comparisons demand for at most $O(n(\log n-\log\sigma_f))$ bit operations. It follows that the total cost at each node is bounded by $\Otilde(n(n\Gamma+\tau+\rho-\log\sigma_f))$ bit operations. The bound on the total cost then follows from our result on the size of the recursion tree.
\end{proof}

\subsection{Known $L_f$ and $\sigma_f$}\label{sec:knownvalues}

From Corollary~\ref{cor:precision}, we already know that, for $\rho\ge \rho_f$, each root $z_i$ of $f$ moves by at most $\frac{\sigma(z_i,f)}{64n^3}$ when passing from $f$ to an arbitrary approximation $\tilde{f}\in [f]_{2^{-\rho_f}}$; see Definition~\ref{def:sufficintlylarge} for the definition of $\rho_f$. Hence, we expect it to be possible to isolate the roots of $f$ by only considering approximations of $f$ (and the intermediate results $f_I$) to $\rho_{f}$ bits after the binary point. The following theorem proves a corresponding result.

\begin{thm}\label{thm:success}
Let $f$ be a polynomial as in (\ref{polyf}) and $\rho\in\N$ an integer with
\begin{align}
\rho\ge \rho_{f}^{\max}:=\left\lceil  \rho_f-3\log\sigma_f+16n\right\rceil=O(\Sigma_f+n).\label{largeenoughL}
\end{align} 
Then, $\dcm^\rho$ returns isolating intervals for \emph{all} roots of $f$ and $B_J>2^{\rho_f}$ for all $(J,s_{J,l},s_{J,r},B_J)\in\mathcal{O}$.
\end{thm}

\begin{proof}
Due to Theorem~\ref{thm:complexitydcm} and~\ref{complexitydcml}, the height $h(\dcm^\rho)$ of $T_{\dcm^\rho}$ is bounded by $$h(\dcm^\rho)\le \log\frac{16n^2}{\sigma_f}=2\log n+4-\log\sigma_f\le 4n-\log\sigma_f.$$ Then, for any interval $I=(a,b)$ produced by $\dcm^\rho$, we have 
\begin{align}
\rho_I\ge \rho+2\log w(I)\ge \rho-2 h(\dcm^\rho)\ge \rho_{f}^{\min}:= \left\lceil \rho_f+8n-\log\sigma_f\right\rceil>0.\label{sizeLI}
\end{align}
The latter inequality guarantees that $\dcm^\rho$ does not return ``insufficient precision''. Now let $I$ be an interval whose closure $\overline{I}$ contains a root $\xi=z_{i_{0}}$ of $f$. We aim to show the following facts:\vspace{0.25cm}  
\begin{enumerate}
\item $I$ is not discarded in Step 3 of $\dcm^\rho$.
\item If $t_{3/2}^{(\tilde{f_I})'}(0,2)>-n2^{n+1-\rho_I}$, then all inequalities (\ref{ineq1})-(\ref{ineq4}) are fulfilled.
\item In the latter case, either $\tilde{I}=(a-\frac{w(I)}{2n},b+\frac{w(I)}{2n})$ is added to $\mathcal{O}$ or $\tilde{I}$ only intersects intervals $J$, with a corresponding $(J,s_{J,l},s_{J,r},B_J)\in\mathcal{O}$, which already isolate $\xi$.\vspace{0.25cm}
\end{enumerate}
If (1)-(3) hold, then $\dcm^\rho$ outputs isolating intervals for \emph{all} real roots of $f$. Namely, $\dcm^\rho$ starts subdividing $I_0=(-\frac{1}{2},\frac{1}{2})$ which contains all real roots of $f$. Thus, for each root $\xi$ of $f$, we eventually obtain an interval $I$ such that $\overline{I}$ contains $\xi$ and $t_{3/2}^{(\tilde{f_I})'}(0,2)>-n2^{n+1-\rho_I}$. Then, either $\tilde{I}$ is added to the list of isolating intervals or $\mathcal{O}$ already contains an isolating interval for $\xi$.

For the proof of (1), we have already shown that $w(I)>\frac{\sigma(\xi,f)}{16n^2}$. Corollary~\ref{cor:precision} then ensures that an arbitrary $g\in [f]_{2^{-\rho_f}}$ has a root $\xi'\in I^+$. Namely, the root $\xi\in\overline{I}$ stays real and moves by at most $\frac{\sigma(\xi,f)}{64n^3}<\frac{w(I)}{4n}$ when passing from $f$ to $g$.
Now, suppose that all coefficients $\tilde{h}_i$ of $\tilde{h}(x)=(1+x)^ n \tilde{f}_{I^+}(\frac{1}{1+x})$ are larger than $-2^{n+2-\rho_I}$; see (\ref{polyh}) for definitions. Since $|h_i-\tilde{h}_i|<2^{n+2-\rho_I}$ for all coefficients $h_i$ of $f_{I^+,\rev}=\sum_{i=0}^n h_i x^i$ (see Lemma~\ref{lem:dcml} (i)), it follows that $h_i>-2^{n+3-\rho_I}$ for all $i$. Hence, for the polynomial $$g(x):=f(x)+2^{n+3-\rho_I}\in [f]_{2^{-\rho_f}},$$ we have $g_{I^+,\rev}(x)=f_{I^+,\rev}(x)+2^{n+3-L_I}(x+1)^n$ and, thus, $g_{I^+,\rev}$ has only positive coefficients. In the case where $\tilde{h}_i<2^{n+2-\rho_I}$ for all $i$, we consider $g(x):=f(x)-2^{n+3-\rho_I}\in [f]_{2^{-\rho_f}}$ and, thus, $g_{I^+,\rho}$ has only negative coefficients. Hence, in both cases, there exists a $g\in [f]_{2^{-\rho_f}}$ which has no root in $I^+$, a contradiction. It follows that $I$ cannot be discarded in Step 3.
 
For (2), suppose that $t_{3/2}^{(\tilde{f_I})'}(0,2)>-n2^{n+1-\rho_I}$. Due to Lemma~\ref{lem:dcml} (i), we have $t_{3/2}^{(f_I)'}(0,2)>-n2^{n+2-\rho_I}$, and since 
$\log\frac{n2^{n+2-\rho_I}}{w(I)}\le 6+3\log n+n-\rho_I-\log\sigma_f
< -\rho_f,$
it follows that $$g(x):=f(x)+x\cdot \frac{n2^{n+2-\rho_I}}{w(I)}\in[f]_{2^{-\rho_f}}.$$
Hence, $g$ has a root $\xi'$ in $I^{+}$. Since $t_{3/2}^{(g_I)'}(0,2)=t_{3/2}^{(f_I)'}(0,2)+n2^{n+2-L_I}>0$, the disc $\Delta_{2w(I)}(a)$ is isolating for $\xi'$.
The following argument shows that $\Delta_{3w(I)/2}(a)$ isolates $\xi$: Suppose that $\Delta_{3w(I)/2}(a)$ contains an additional root $z_j\neq \xi$ of $f$. Then, $\sigma(\xi,f)<3w(I)$ and, thus, $\xi$ and $z_j$ would move by at most $\frac{3w(I)}{64n^3}<\frac{w(I)}{2}$ when passing from $f$ to $g$. It follows that $g$ would have at least two roots within $\Delta_{2w(I)}(a)$, a contradiction. Now, since $\Delta_{3w(I)/2}(a)$ is isolating for $\xi\in\overline{I}$, we have
$$
\frac{\sigma(\xi,f)}{16n^2}<w(I)<2\sigma(\xi,f).
$$
The left inequality implies that the distance of $\xi$ to any of the points $a^+=a-\frac{w(I)}{4n}$, $b^+=b+\frac{w(I)}{4n}$ and $c:=a-\frac{w(I)}{n}$ is larger than or equal to $\frac{w(I)}{4n}>\frac{\sigma(\xi,f)}{64n^3}$. Let $d_i:=|z_i-a|$ denote the distance between a root $z_i\neq \xi$ and the disc $\Delta_{3w(I)/2}(a)$. Then, $$\frac{\sigma(z_i,f)}{64n^3}\le\frac{|z_i-\xi|}{64n^3}\le\frac{d_i+3w(I)}{64n^3}<d_i+\frac{w(I)}{4}.$$ 
It follows that the points $a^+$, $b^+$, $c\in\Delta_{5w(I)/4}(a)$ are located outside the disc $\Delta_i:=\Delta_{\sigma(z_i,f)/(64n^3)}(z_i)$. In summary, none of the discs $\Delta_i$, $i=1,\ldots,n$, contains any of the points $a^+$, $b^+$ and $c$. Hence, due to Corollary~\ref{cor:precision}, it follows that each of the values $|f(c)|$, $|f(a^+)|$ and $|f(b^+)|$ is larger than $(n+1)2^{-\rho_f}$. A simple computation now shows that $(n+1)2^{-\rho_f}>2^{2n+8-\rho_I}n^2$. Thus, according to Lemma~\ref{lem:dcml} (ii), each of the absolute values $|\lambda|$, $|\lambda^-|$ and $|\lambda^+|$ is larger than \begin{align}
(n+1)2^{-\rho_f}-2^{n+3-\rho_I}n>2^{2n+8-\rho_I}n^2-2^{n+3-\rho_I}n>2^{2n+7-\rho_I}n^2.\label{absvalues}
\end{align}
It follows that the inequalities (\ref{ineq3}) and (\ref{ineq4}) hold. Since $I^+$ is isolating for $\xi$, $f(a^+)$ and $f(b^+)$ must have different signs and, thus, the same holds for $\lambda^-$ and $\lambda^+$. Hence, the inequality (\ref{ineq1}) holds as well. In addition, we have $B_{\tilde{I}}=\min(|\lambda^-|,|\lambda^+|)-2^{n+3-\rho_I}n>2^{-\rho_f}$ because of (\ref{absvalues}).
It remains to show (3): If $t_{3/2}^{(\tilde{f_I})'}(0,2)>-n2^{n+1-\rho_I}$, then due to (2) and Lemma~\ref{lem:dcml} (ii), the interval $\tilde{I}$ and the $\frac{w(I)}{n}$-neighborhood of $I$ is isolating for $\xi$. If $\tilde{I}$ does not intersect any other interval in $\mathcal{O}$, then $\tilde{I}$ is added to $\mathcal{O}$ and, thus, $\dcm^\rho$ outputs an isolating interval for $\xi$. We still have to consider the case where $\tilde{I}$ intersects an interval $J$ from $\mathcal{O}$. From the construction of $\mathcal{O}$, $J$ is the extension $(\tilde{c},\tilde{d})$ of an interval $J'=(c,d)$.
Now, suppose that $J$ is isolating for a root $\gamma\neq \xi$. The roots $\xi$ and $\gamma$ move by at most $\frac{w(I)}{4n}$ and $\frac{w(J')}{4n}$, respectively, when passing from $f$ to an arbitrary $g\in [f]_{2^{-\rho_f}}$ (see the proof of (1)). Hence, it follows that the union of $(a-w(I),b+w(I))$ and $(c-w(J'),d+w(J'))$ contains at least two roots of any $g\in [f]_{2^{-\rho_f}}$. Due to Lemma~\ref{lem:intersection}, one of the discs $\Delta_{2w(I)}(a)$ or $\Delta_{2w(J')}(c)$ then also contains at least two roots of $g$ contradicting the fact that $t_{3/2}^{(p_I)'}(0,2)>0$ for $p(x):=f(x)+x\cdot \frac{n2^{n+2-\rho_I}}{w(I)}\in [f]_{2^{-\rho_f}}$ and $t_{3/2}^{(q_{J'})'}(0,2)>0$ for $q(x):=f(x)+x\cdot \frac{n2^{n+2-\rho_{J'}}}{w(J')}\in [f]_{2^{-\rho_f}}$.
It follows that $J$ already isolates $\xi$.
\end{proof}

\subsection{Unknown $\rho_f$ and $\sigma_f$}\label{sec:certify}

For unknown $\rho_f$ and $\sigma_f$, we proceed as follows: We start with an initial precision $\rho$ (e.g., $\rho=16$) and run $\dcm^\rho$. If $\dcm^{\rho}$ returns "insufficient precision", we double $\rho$ and start over. Otherwise, $\dcm^{\rho}$ returns a list $\mathcal{O}=\{(J_k,s_{k,l},s_{k,r},B_k)\}_{k=1,\ldots,m}$, where each interval $J_k=(c_k,d_k)$ isolates a real root of $f$, $s_{k,l}=\sgn f(c_k)$, $s_{k,r}=\sgn f(d_k)$ and $0<B_{k}<\min(| f(c_k)|,| f(d_k)|)$. As already mentioned, there is no guarantee that all roots of $f$ are captured. Hence, in a second step, we use the subsequently described method $\textsc{Certify}^{\rho}$ to check whether the \emph{region of uncertainty}
$$\mathcal{R}:=\left[-\frac{1}{2},\frac{1}{2}\right]\backslash\bigcup_{k=1}^{m} J_k$$
may contain a root of $f$. If we can guarantee that $f(x)\neq 0$ for all $x\in\mathcal{R}$, we return the list $\mathcal{L}=\{J_k\}_{k=1,\ldots,m}$ of isolating intervals. Otherwise, we double $\rho$ and start over the entire algorithm. We have already proven in Theorem~\ref{thm:success} that $\dcm^{\rho}$ isolates all real roots of $f$ if $\rho\ge\rho_f^{\max}$ (i.e., $\rho$ fulfills the inequality (\ref{largeenoughL})). The following considerations will show that, for $\rho\ge\rho_f^{\max}$, $\textsc{Certify}^\rho$ succeeds as well.\\
 
How can we guarantee that $f$ does not vanish on $\mathcal{R}$? The crucial idea is to consider a decomposition of $[-\frac{1}{2},\frac{1}{2}]$ into subintervals $I$ and corresponding $\mu_I$-approximations $g$ of $f_{I}$ such that $g$ is monotone on $[0,1]$ or $\T_{3/2}^{g}(0,1)$ holds. Namely, for such an interval $I$, we can easily estimate the image $g([0,1])$ and, thus, conclude that $f$ contains no root in $I\cap \mathcal{R}$ or $\rho<\rho_f^{\max}$ because $g(t)$ and $f_I(t)$ differ by at most $(n+1)\mu_I$ for all $t\in[0,1]$. More precisely, we have:

\begin{lem}\label{lem:cert}
Let $I=(a,b)$ be an interval and $g(x)$ a $\mu$-binary approximation of $f_{I}$ with 
\begin{align}
-\log\mu\ge \rho-2(4n-\log\sigma_f).\label{condition}
\end{align}
\begin{itemize}
\item[(i)] Suppose that $\T_{3/2}^{g}(0,1)$ holds and $I$ is not entirely contained in one of the $J_k$. If \begin{align}
|g(0)|>8n\mu,\label{test1}
\end{align}
then $\bar{I}$ contains no root of $f$. Otherwise, $\rho<\rho_f^{\max}$. 
\item[(ii)] Suppose that $g$ is monotone on $[0,1]$ and let
$
\bar{I}\cap\mathcal{R}=\bigcup_{i=1}^s L_i
$
be the intersection of $\bar{I}$ and $\mathcal{R}$. For each endpoint $q$ of an arbitrary $L_i$, we define
\begin{align}
\lambda(q):=\begin{cases}
  s_{k,l}\cdot B_{k},  & \text{if }q\notin\{a,b\}\text{ and }q\text{ is the left endpoint of an interval } J_k\\
  s_{k,r}^r\cdot B_{k},  & \text{if }q\notin\{a,b\}\text{ and }q\text{ is the right endpoint of an interval } J_k\\
  g(0),  & \text{if }q=a\\
  g(1),  & \text{if }q=b.
\end{cases}\label{lambdaq}
\end{align}
If, for all $L_i=[q_l,q_r]$, $\min(|\lambda(q_l)|,|\lambda(q_r)|)>4n\mu$ and $\lambda(q_l)\cdot \lambda(q_r)>0$, then $\bar{I}\cap\mathcal{R}$ contains no root of $f$. Otherwise, we have $\rho<\rho_f^{\max}$.
\end{itemize}
\end{lem}

\begin{figure}[t]
\begin{center}
\vspace{-12.2cm}\includegraphics[width=14cm]{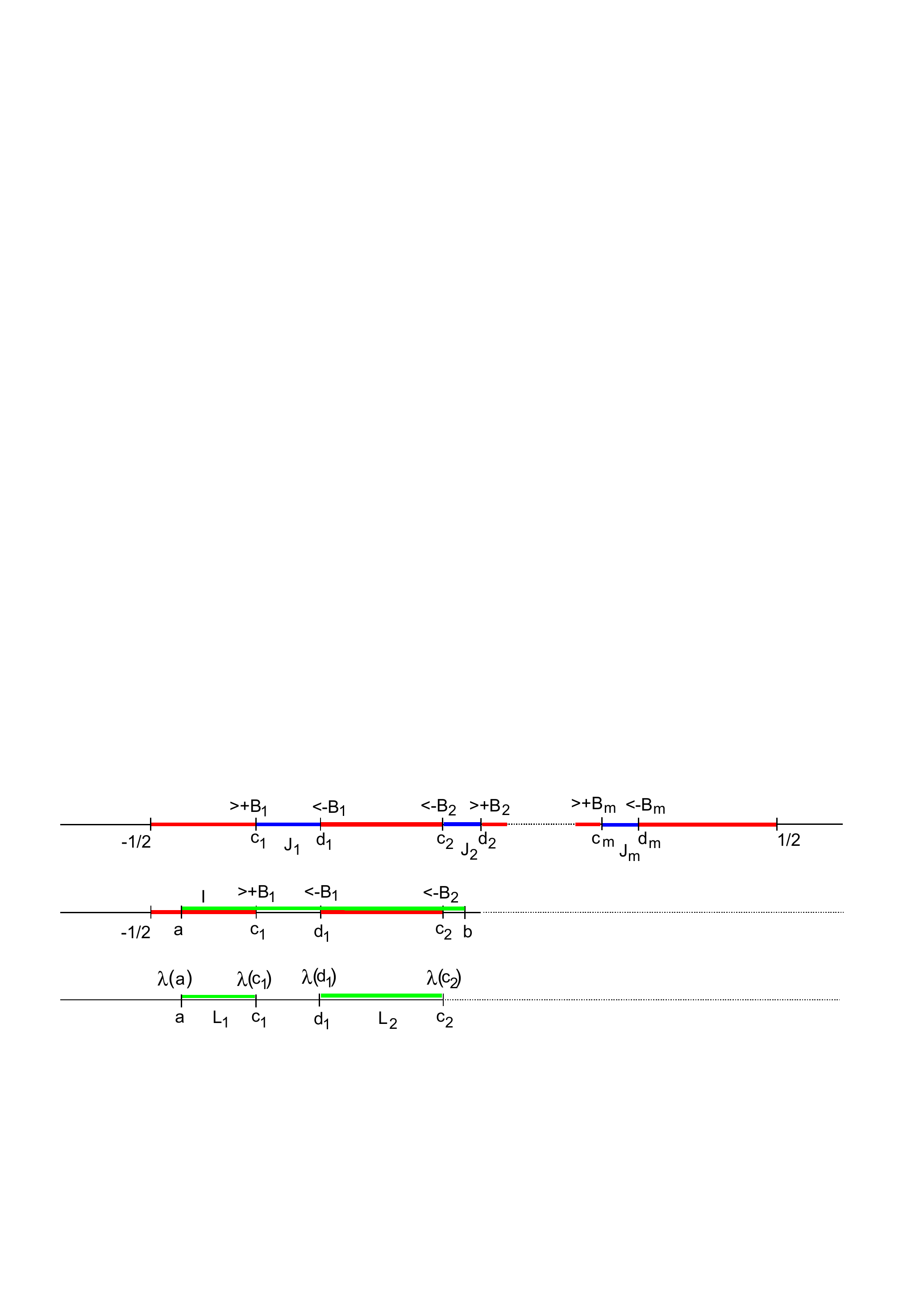}\end{center}
\vspace{-4cm}\caption{\label{fig:certify} \footnotesize{$\dcm^\rho$ returns a list $\mathcal{O}=\{J_k,s_{k,l},s_{k,r},B_k\}_k$, where $J_k$ is isolating for a real root of $f$, $s_{k,l}=\sgn f(c_k)$, $s_{k,r}=\sgn f(d_k)$ and $\min(|f(c_k)|,|f(d_k)|)>B_k>0$.
The intervals in between define the \emph{region of uncertainty} $\mathcal{R}$. In $\textsc{Certify}^\rho$, we subdivide $(-1/2,1/2)$ into intervals $I$ such that, for a $\mu$-approximation $g$ of $f_I$, either $\T_{3/2}^{g}(0,1)$ holds or $g$ is monotone on $(0,1)$. If $\T_{3/2}^{g}(0,1)$ holds and $|g(0)|>8m\mu$, then $I$ contains no root of $f$; see Lemma~\ref{lem:cert} (i). If $g$ is monotone on $(0,1)$, we consider all intervals $L_i$ in the intersection of $I$ with $\mathcal{R}$ and check whether the conditions in Lemma~\ref{lem:cert} (ii) are fulfilled. If they are fulfilled, then $f$ has no root in $L_i$; otherwise, we must have $\sigma<\sigma_f^{\max}$.}}
\end{figure}

\begin{proof}
If $\T_{3/2}^{g}(0,1)$ holds, then $\frac{1}{3}|g(0)|<|g(t)|<\frac{5}{3}|g(0)|$ for all $t\in[0,1]$ according to Lemma~\ref{lem:test}. It follows that
$$
|f_I(t)|\ge |g(t)|-|g(t)-f_I(t)|> \frac{1}{3}|g(0)|-|g(t)-f_I(t)|>\frac{8}{3}n\mu-(n+1)\mu>0, 
$$
hence, $f$ has no root in $\bar{I}$. Now suppose that $|g(0)|\le 8n\mu$. Since $I$ is not contained in any $J_k$, there exists a $t\in[0,1]$ with $x=a+t(b-a)\in\mathcal{R}$ and $|f(x)|=|f_I(t)|\le |g(t)|+(n+1)\mu\le\frac{5}{3}|g(0)|+(n+1)\mu<16n\mu$. If $\rho\ge\rho_f^{\max}$, then from (\ref{condition}) and the definition of $\rho_f^{\max}$, it follows that $-\log\mu\ge \rho_f^{\min}=\left\lceil \rho_f+8n-\log\sigma_f\right\rceil$; see the computation in (\ref{sizeLI}). Hence, we have $|f(x)|<2^{-\rho_f}$. In addition, Lemma~\ref{lem:dcml} (iv) and Theorem~\ref{thm:success} guarantee that $\dcm^\rho$ returns isolating intervals for all real roots of $f$, and each point in $\mathcal{R}$ has distance $\ge\sigma(z_i,f)/(64n^3)$ from each root $z_i$. Thus, $|f(x)|>(n+1)2^{-\rho_f}$ due to Corollary~\ref{cor:precision}, a contradiction. This proves (i).
 
For (ii), we consider an arbitrary interval $L_i=[q_l,q_r]$. Let $t_l$ and $t_r$ be corresponding values in $[0,1]$ with 
$q_l=a+t_l\cdot w(I)$ and $q_r=a+t_r\cdot w(I)$. If 
$\min(|\lambda(q_l)|,|\lambda(q_r)|)>4n\mu$, then 
$$\min(|g(t_l)|,|g(t_r)|)\ge\min(|\lambda(q_l)
|,|\lambda(q_r)|)-(n+1)\mu>2n\mu.$$ Namely, for $q_l=a$, we obviously have $|g(t_l)|=|\lambda(q_l)|$; otherwise, 
$|g(t_l)|\ge |f_I(t_l)|-(n+1)\mu\ge 
|\lambda(q_l)|-(n+1)\mu$. For $q_r$, an analogous argument applies. 
If, in addition, $\lambda(q_l)\cdot \lambda(q_r)>0$, then $g(t_l)\cdot g(t_r)>0$ as well because $\lambda(q_l)$ and $\lambda(q_r)$ have the same sign as $g(t_l)$ and $g(t_r)$, respectively. Since we assumed 
that $g$ is monotone on $[0,1]$, it follows that $|g(t)|>2n\mu$ for all $t\in[t_l,t_r]$. This shows that 
$|f_I(t)|\ge |g(t)|-(n+1)\mu>0$ for all $t\in[t_l,t_r]$, thus the first part of (ii) follows. For the second part, suppose that 
$\rho\ge\rho_f^{\max}$. Then, $B_k>2^{-\rho_f}>4n\mu$ for all $k$ and 
$|f(x)|>2^{-\rho_f}(n+1)$ for all $x\in\mathcal{R}$ according to Corollary~\ref{cor:precision} and Theorem~\ref{thm:success}. Thus, if 
$a\in\mathcal{R}$, we have $$|g(0)|\ge 
|f_I(0)|-(n+1)\mu=|f(a)|-(n+1)\mu>2^{\rho_f}
-(n+1)\mu>4n\mu.$$ An analogous argument applies to $b$. It 
follows that $|\lambda(q)|>4n\mu$ for all endpoints $q$ of an arbitrary interval 
$L_i=[q_l,q_r]$. It remains to show that $\lambda(q_l)\cdot\lambda(q_r)>0$. We have already shown that $|\lambda(q)|>4n\mu$ for 
each endpoint $q$, thus, $f(q)$ must have the same sign as $\lambda(q)$. 
Namely, if $q\in\{a,b\}$, then $f(q)$ differs from $\lambda(q)>4n\mu$ 
by at most $(n+1)\mu<4n\mu$, and, for $q\notin\{a,b\}$, we have 
$\sgn(\lambda(q))=s_{k,l}$ or $\sgn(\lambda(q))=s_{k,r}$ depending on 
whether $q$ is the left or the right endpoint of an interval $J_k$. Since 
$\rho\ge\rho_f^{\max}$, $\mathcal{R}$ contains no root of $f$, thus, we 
must have $\lambda(q_l)\cdot\lambda(q_r)=f(q_l)\cdot f(q_r)>0$.
\end{proof}

We can now formulate the subroutine $\textsc{Certify}^\rho$ (see Algorithm~\ref{alg:certify} in the Appendix for pseudo-code). $\textsc{Certify}^\rho$ is similar to $\dcm^\rho$ in the sense that we recursively subdivide $I_0=(-\frac{1}{2},\frac{1}{2})$ into intervals $I$ and consider corresponding $\rho_I$-binary approximations $\tilde{f}_I$ of $f_I$. Then, in each iteration, we aim to apply Lemma~\ref{lem:cert} in order to certify that $\bar{I}\cap\mathcal{R}$ contains no root of $f$ or $\rho<\rho_f^{\max}$. Throughout the following consideration, we \emph{assume} that 
\begin{align}
\textsc{Certify}^{\rho}\text{ never produces an interval }I\text{ of width } w(I)\le \frac{\sigma_f}{8n^2}.\label{assumption}
\end{align} 
We will prove this fact in Theorem~\ref{thm:certify} (ii). Again, we mark comments which should help to follow the approach by an "//" at the beginning. 

\vspace{0.5cm}

\hrule\vspace{0.2cm}
\noindent \textbf{$\textsc{Certify}^\rho$.} In a first step, we choose a $(\rho+n+1)$-binary approximation $\tilde{f}$ of $f$ and evaluate $\tilde{f}(-\frac{1}{2}+x)$. Then, the resulting polynomial is approximated by a $(\rho+1)$-binary approximation $\tilde{f}_{I_0}\in [\tilde{f}(-\frac{1}{2}+x)]_{2^{-\rho-1}}$, thus, $\tilde{f}_{I_0}\in [f_{I_0}]_{2^{-\rho}}$ according to Lemma~\ref{lem:taylorshift}.\\ 

$\textsc{Certify}^\rho$ maintains a list $\mathcal{A}$ of active nodes $(I,\tilde{f}_I,\rho_I)$, where $I=(a,b)\subset I_0$ is an interval, $\tilde{f}_I$ approximates $f_I$ to $\rho_I$ bits after the binary point and $\rho+2\log w(I)\le\rho_I\le \rho$. We initially start with $\mathcal{A}:=\{(I_0,\tilde{f}_{I_0},\rho)\}$. For each active node, we proceed as follows:\vspace{0.25cm}
\begin{enumerate} 
\item Remove $(I,\tilde{f}_I,\rho_I)$ from $\mathcal{A}$.
\item If $I\cap \mathcal{R}=\emptyset$, do nothing (i.e., discard $I$). Otherwise, compute $t_{3/2}^{\tilde{f_{I}}}(0,1)$.\\

\noindent$\text{\large{//}}$\textit{ If $I\cap \mathcal{R}=\emptyset$, $I$ is contained in one of the isolating intervals $J_k$, hence, we can discard $I$.}\\

\item If $t_{3/2}^{\tilde{f_{I}}}(0,1)>-2^{-\rho_{I}+2}n$, check whether
\begin{align}
|\tilde{f}_{I}(0)+2^{-\rho_I+2}n|>2^{-\rho_I+5}n^2.\label{check1}
\end{align} 
If (\ref{check1}) holds, do nothing (i.e., discard $I$); otherwise, return "insufficient precision".\\

\noindent$\text{\large{//}}$\textit{ For $g(x):=f_{I}(0)+2^{-\rho_I+2}n\in [f_I]_{2^{-\rho_I+3}n}$, the predicate $\T_{3/2}^g(0,1)$ holds. From our assumption on $w(I)$, we further have $\rho_I\ge\rho+2\log w(I)\ge \rho-2(3+2\log n-\log\sigma_f)$, and, thus, $2^{-\rho_I+3}n\le 2^{-\rho-2(4n-\log\sigma_f)}$. It follows that $g$ fulfills the condition (\ref{condition}) from Lemma~\ref{lem:cert} and, therefore, $\bar{I}$ contains no root of $f$ if (\ref{check1}) holds; otherwise, $\sigma<\sigma_f^{\max}$.}\\

\item If $t_{3/2}^{\tilde{f_{I}}}(0,1)\le -2^{-\rho_{I}+2}n$, compute $\tilde{h}(x)=\sum_{i=0}^{n}\tilde{h}_{i}x^{i}:=(1+x)^{n}(\tilde{f_I})'(\frac{1}{1+x})$ and consider the following distinct cases:\\

\begin{itemize}
\item[(a)] If $\tilde{h}_{i}>-n2^{n-\rho_{I}}$ for all $i$ (or $\tilde{h}_{i}<n2^{n-\rho_{I}}$ for all $i$), consider $$g(x):=\tilde{f}_I(x)+n2^{n-\rho_I}\cdot x\in [f_I]_{n2^{n-\rho_I}}$$
($g(x):=\tilde{f}_I(x)+n2^{n-\rho_I}\cdot x$, respectively). Then, for each interval $L_i=[q_l,q_r]$, determine $\lambda(q_l)$ and $\lambda(q_r)$ as defined in (\ref{lambdaq}). If $\min(|\lambda(q_l)|,|\lambda(q_r)|)>n2^{n+2-\rho_I}$ and $\lambda(q_l)\cdot\lambda(q_r)>0$ for all $L_i$, discard $I$; otherwise, return "insufficient precision".\\

\noindent$\text{\large{//}}$\textit{ Suppose $\tilde{h}_{i}>-n2^{n-\rho_{I}}$ for all $i$ and $g(x):=\tilde{f}_I(x)+n2^{n-\rho_I}x$. Then, the polynomial $(1+x)^{n}(g)'(\frac{1}{1+x})=(1+x)^{n}(\tilde{f_I})'(\frac{1}{1+x})+n2^{n-\rho_I}(1+x)^n$ has only positive coefficients. It follows that $\var(g',(0,1))=0$ and, therefore $g$ is monotone on $[0,1]$. In addition, from our assumption on $w(I)$, we have $n2^{n-\rho_I}\le 2^{-\rho-2(4n-\log\sigma_f)}$. Hence, we can apply Lemma~\ref{lem:cert} (ii) to $g$ which guarantees that $\bar{I}\cap\mathcal{R}$ does not contain a root of $f$ if $\min(|\lambda(q_l)|,|\lambda(q_r)|)>n2^{n+2-\rho_I}$ and $\lambda(q_l)\cdot\lambda(q_r)>0$ for all $L_i=[q_l,q_r]$. If one of the latter two inequalities does not hold, then $\sigma<\sigma_f^{\max}$. The case $\tilde{h}_{i}<n2^{n-\rho_{I}}$ for all $i$ is treated in exactly the same manner.}\\

\item[(b)] If there exist $\tilde{h}_i$ and $\tilde{h}_j$ with $\tilde{h}_i\le-n2^{n-\rho_{I}}$ and $\tilde{h}_j\ge n2^{n-\rho_{I}}$, then $I$ is subdivided into $I_l:=(a,m_I)$ and $I_r:=(m_I,b)$. We add $(I_l,\tilde{f_{I_l}},\rho_I-1)$ and
$(I_r,\tilde{f_{I_r}},\rho_I-2)$ to $\mathcal{A}$, where $\tilde{f_{I_l}}$ is an $\rho_I$-binary approximation of $\tilde{f_{I}}(\frac{x}{2})$ and $\tilde{f_{I_r}}$ an $(\rho_I-1)$-binary approximation of $\tilde{f_I}(\frac{x+1}{2})$; see Step 4 (b) of $\dcm^\rho$ for details. If $\rho_I<2$, return "insufficient precision".\\

\noindent$\text{\large{//}}$\textit{ Due to Lemma~\ref{lem:taylorshift}, we have $\tilde{f}_{I_l}\in [f_{I_l}]_{2^{-\rho_I-1}}$ and $\tilde{f}_{I_r}\in [f_{I_r}]_{2^{-\rho_I-2}}$. Hence, by induction, it follows that $\rho+2\log w(I)\le\rho_I\le \rho$ for all active nodes.}\\ 

\end{itemize}
\end{enumerate}
$\textsc{Certify}^\rho$ stops when $\mathcal{A}$ becomes empty. If $\textsc{Certify}^{\rho}$ returns "insufficient precision", we know for sure that $\sigma<\sigma_f^{\max}$. Otherwise, the region of uncertainty $\mathcal{R}$ contains no root of $f$.\vspace{0.2cm}
\vspace{0.2cm}
\hrule\vspace{0.5cm}

The following theorem proves that our assumption (\ref{assumption}) for the intervals produced by $\textsc{Certify}^{\rho}$ is correct. Furthermore, we show that $\textsc{Certify}^{\rho}$ is also efficient with respect to bit complexity matching the worst case bound obtained for $\dcm^{\rho}$; see Theorem~\ref{complexitydcml}.

\begin{thm}\label{thm:certify}
For a polynomial $f$ as defined in (\ref{polyf}) and an arbitrary $\rho\in\N$, 
\begin{itemize}
\item[(i)] $\textsc{Certify}^\rho$ does not produce an interval $I$ of width $w(I)\le\sigma_f/(8n^2)$ and induces a recursion tree of size $O(\Sigma_f+n\log n)$.
\item[(ii)] $\textsc{Certify}^\rho$ needs no more than $\Otilde(n(\Sigma_f+n\log n)(n\Gamma+\tau+\rho-\log\sigma_f))$ bit operations.  
\item[(iii)] For $\sigma\ge\sigma_f^{\max}$, $\textsc{Certify}^\rho$ succeeds.
\end{itemize}
\end{thm}

\begin{proof}
An interval $I$ is only subdivided if $t_{3/2}^{\tilde{f_{I}}}(0,1)\le -2^{-\rho_{I}+2}n$ (Step (3)) or if there exist coefficients $\tilde{h}_i$ and $\tilde{h}_j$ of $\tilde{h}(x)=\sum_{i=0}^n \tilde{h}_i x^i=(1+x)^n(\tilde{f_I})'(\frac{1}{1+x})$ with $\tilde{h}_i<-n2^{n-\rho_I}$ and $\tilde{h}_j>n2^{n-\rho_I}$ (Step 4 (b)). In the first case, we must have $t_{3/2}^{f_{I}}(0,1)<0$ since $|t_{3/2}^{f_I}(0,1)-t_{3/2}^{\tilde{f_{I}}}(0,1)|<2^{-\rho_I+2}n$, hence, $T_{3/2}^{f_I}(0,1)$ does not hold. For the second case, we have $\var((f_I)',(0,1))\neq 0$ since corresponding coefficients of $\tilde{h}(x)=(1+x)^n(\tilde{f_I})'(\frac{1}{1+x})$ and $(1+x)^n(f_I)'(\frac{1}{1+x})$ differ by at most $n2^{n-\rho_I}$, thus, $\var(f',I)=\var((f')_I,(0,1))=\var((f_I)',(0,1))\neq 0$. Hence, 
the first part of (i) follows from Lemma~\ref{lem:success} (iv) and (iii) is then immediate from the remarks in the above description of $\textsc{Certify}^\rho$. Namely, $\textsc{Certify}^\rho$ only returns "insufficient precision" if $\rho<\rho_f^{\max}$ and guarantees that $f(x)\neq 0$ for all $x\in\mathcal{R}$, otherwise.
For the second part of (i), we remark that, due to the above argument, an interval $I$ is terminal if the disc $\Delta_{2nw(I)}(m(I))$ does not contain a root $\xi$ of $f$ with $\sigma(\xi,f)<4n^2 w(I)$. In~\cite[Section 4.2]{s-bitstream-mcs11}, it is shown that the recursion tree $T(f')$ induced by the latter property~\footnote{In~\cite[Section 4.2]{s-bitstream-mcs11}, $T(f')$ is defined as subdivision tree obtained by recursive bisection of the interval $(-\frac{1}{4},\frac{1}{4})$ in accordance with the following rule: At depth $h\in\N_0$, an interval $I=(-\frac{1}{4}+i2^{-h-1},-\frac{1}{4}+(i+1)2^{-h-1})$ is subdivided if and only if $\var(f',I)\neq 0$ and $\Delta_{2^8n^5w(I)}(m(I))$ contains a root $\xi$ of $f$ with separation $\sigma(\xi,f)<2^7n^5w(I)$. For the given situation, we can alternatively define $T(f')$ as the (even smaller) tree obtained by recursive bisection of $(-\frac{1}{2},\frac{1}{2})$, where an interval $I$ is subdivided if $\var(I,f')\neq 0$ and $\Delta_{2nw(I)}(m(I))$ contains a root $\xi$ of $f$ with $\sigma(\xi,f)<4n^2 w(I)$.} has size $O(\Sigma_f+n\log n)$. Hence, the same holds for the recursion tree induced by $\textsc{Certify}^\rho$ which is a subtree of $T(f')$. Finally, (iii) follows in completely analogous manner as the result on the bit complexity for $\dcm^\rho$ as shown in the proof of Theorem~\ref{complexitydcml}. 
\end{proof}

Eventually, we present our overall root isolation method $\R\textsc{Isolate}$. It applies to a polynomial $F$ as given in (\ref{polyF}) and returns isolating intervals for all real roots of $F$.\vspace{0.5cm}

\hrule\vspace{0.2cm}
\textbf{$\R\textsc{Isolate}$:} Choose a starting precision $\rho\in\mathbb{N}$ (e.g., $L=16$) and run $\dcm^\rho$ on the polynomial $f$ as defined in (\ref{polyf}). If $\dcm^\rho$ returns ``insufficient precision'', we double $\rho$ and start over again. Otherwise, $\dcm^\rho$ returns a list $\mathcal{O}=\{(J_k,s_{k,l},s_{k,r},B_{k})\}_{k=1,\ldots,m}$ with isolating intervals $J_k$ for some of the real roots of $f$. If $\textsc{Certify}^\rho$ returns ``insufficient precision'', we double $\rho$ and start over the entire algorithm. If $\textsc{Certify}^\rho$ succeeds, the intervals $J_k=(c_k,d_k)$ isolate all real roots of $f$. Hence, we return the intervals $(2^{\Gamma+1} c_k, 2^{\Gamma+1} d_k)$, $k=1,\ldots,m$, which isolate the real roots of $F$.\vspace{0.2cm}
\hrule\vspace{0.4cm}

The following theorem summarizes our results:

\begin{thm}
Let $F$ be a polynomial as given in (\ref{polyF}). Then, $\R\textsc{Isolate}$ determines isolating intervals for all real roots of $F$ and, for each of these intervals $J$ containing a root $\xi$ of $F$, it holds that $$\frac{\sigma(\xi,F)}{16n^2}<w(J)<2n\sigma(\xi,F).$$ $\R\textsc{Isolate}$ demands for coefficient approximations of $F$ to $\tilde{O}(\Sigma_F+n\Gamma_F)$ bits after the binary point and the total cost is bounded by
$$
\Otilde(n(\Sigma_F+n\Gamma_F)^2)=\Otilde(n(\Sigma_F+n\tau)^2)
$$
bit operations. For $F\in\Z[x]$, the bound on the bit complexity writes as $\Otilde(n^3\tau^2)$.
\end{thm} 

\begin{proof}
It remains to prove the complexity bounds and the claim on the width of the isolating intervals. According to Appendix~\ref{sec:rootbound}, the computation of an approximate logarithmic root bound $\Gamma\in\N$ as defined in Section~\ref{sec:scaling} amounts for $\tilde{O}((n\Gamma_F)^2)$ bit operations. For a certain precision $\rho$, the total cost for running $\dcm^\rho$ and $\textsc{Certify}^\rho$ is bounded by 
$$\Otilde(n(\Sigma_f+n\log n)(n\Gamma+\tau+\rho-\log\sigma_f))=\Otilde(n(\Sigma_F+n\Gamma)(n\Gamma+\tau+\rho-\log\sigma_F))$$ 
bit operations; see Theorem~\ref{complexitydcml} and Theorem~\ref{thm:certify}. Since 
we double $\rho$ in each step and succeed for $\rho\ge \rho_f^{\max}$, $\rho$ is always bounded by $2\rho_f^{\max}=O(\Sigma_f+n)=O(\Sigma_F+n\Gamma)$. It follows that the total costs are dominated by the cost for the last run which is
$
\Otilde(n(\Sigma_F+n\Gamma)(n\Gamma+\tau+\Sigma_F)).
$
Furthermore, we have to approximate the coefficients of $f$ to $O(\Sigma_f+n)=O(\Sigma_F+n\Gamma)$ bits after the binary point. Hence, the coefficients of $F$ have to be approximated to $O(\Sigma_F+n\Gamma+\tau)$ bits after the binary point; see Section~\ref{sec:apx} for more details. From our construction of $f$ and $\Gamma$, it holds that $\Gamma<4\log n+\Gamma_F$ and $\tau=\tilde{O}(n\Gamma_F)$ (see Appendix~\ref{sec:rootbound}), hence, we can replace $\Gamma$ by $\Gamma_F$ and further omit $\tau$ in the above complexity bounds.
For the special case where $F$ is an integer polynomial, the bound on the bit complexity follows from $\Sigma_F=\Otilde(n\tau)$; see Appendix~\ref{integer}. The estimate on the size of the isolating intervals is due to the following consideration: An interval $I$ which contains the root $z=\frac{\xi}{2^{\Gamma+1}}$ of $f$ is not subdivided by $\dcm^\rho$ if $w(I)\le\sigma(z,f)/(8n^2)$. Hence, any interval $J_k$ which is returned by $\dcm^\rho$ as an isolating interval for $z$ is the extension $\tilde{I}=(a-\frac{w(I)}{2n},b+\frac{w(I)}{2n})$ of an interval $I=(a,b)$ with $w(I)>\sigma(z,f)/(16n^2)$, thus $w(J)=2^{\Gamma+1} w(I)>\sigma(\xi,F)/(16n^2)$. From our construction, the $\frac{w(I)}{n}$-neighborhood of $I$ isolates $z$ as well and, thus, $w(J)=2^{\Gamma+1} w(I)<4\Gamma n\sigma(z_i,f)=2n\sigma(\xi,F)$.
\end{proof}

\subsection{Some Remarks}

\subsubsection{On the Complexity Analysis for Integer Polynomials}

We remark that in order to achieve the complexity bound $\tilde{O}(n^3\tau^2)$ for integer polynomials, the subroutine $\textsc{Certify}^{\rho}$ and its analysis is not needed. Namely, due to our considerations in Appendix~\ref{integer}, we can compute upper bounds for $\Sigma_f$ (in terms of $n$ and $\tau$) and, thus, also an upper bound $\rho^*(n,\tau)$ for $\rho_f^{\max}$ which matches $\rho_f^{\max}$ at least with respect to worst case complexity. Then, according to Theorem~\ref{thm:success}, it is guaranteed that $\dcm^{\rho^*(n,\tau)}$ computes isolating intervals for all real roots of $f$. Unfortunately, this approach cannot be considered practical at all because such upper bounds usually tend to be much larger than the actual $\rho_f^{\max}$. We would like to emphasize on the fact that our algorithm is output sensitive in the way that it demands for a precision which is not much larger than $\rho_f$, hence, our algorithm chooses an almost optimal precision. Without giving an exact mathematical proof, we conjecture that, for any bisection method, the bound on the bit complexity as achieved by our algorithm is optimal (up to $\log$-factors). Namely, the bound $\tilde{O}(n\tau)$ on the precision as well as the bound $\tilde{O}(n\tau)$ on the size of the recursion tree cannot be lowered for Mignotte polynomials.

\subsubsection{On Efficient Implementation}
We formulated our algorithm in a way to make it accessible to the complexity analysis but still feasible and efficient for an implementation. Nevertheless, we recommend to consider a slight modification of our algorithm when actually implementing it. 

For our certification step $\textsc{Certify}$, the most obvious modification is to only subdivide the region $\mathcal{R}$ instead of the entire interval $(-\frac{1}{2},\frac{1}{2})$. More precisely, $\mathcal{R}$ decomposes into intervals $L_j$ "in between" the isolating intervals $J_k$. Then, we approximate the polynomials $f_{L_j}$ to $\rho$ bits after the binary point and recursively proceed each $L_j$ in a similar way as proposed in $\textsc{Certify}^{\rho}$. An experimental implementation of our algorithm in \textsc{Maple} has shown that following this approach the running time for the certification step is almost negligible whereas, for the original formulation, it is approximately of the same magnitude as the running time for $\dcm^{\rho}$. 

Furthermore, we propose to also use the inclusion predicate based on Descartes' Rule of Signs. With respect to complexity, our inclusion predicate based on the $\T'_{3/2}$-test (see Corollary~\ref{inclusion}) is comparable to Descartes' Rule of Signs, where we check whether $f$ has exactly one sign variation for a certain interval. However, in practice, this subtle difference is crucial because already $\log n$ bisection steps more for each root may render an algorithm inefficient. As an alternative, for each interval $I$ in the subdivision process, we propose to check whether there exists a suitable $\rho_I$-approximation $g$ of $f_I$ with $\var(g,(0,1))=1$. Namely, if there exists such a $g$, then we can proceed with $\tilde{f_I}:=g$ which has exactly one root in $I$. Thus, it is easy to refine $I$ (via simple bisection or quadratic interval refinement) such that $\T^{g'}_{3/2}(0,2)$ holds as well.

We finally report on an interesting behavior of the proposed method. It is easy to see that, for small intervals $I=(a,b)$, the leading coefficients of $f_I(x)=f(a+w(I)x)$ are considerably smaller than the first-order coefficients. Since we only consider a certain number $\rho_I\le\rho$ of bits after the binary point, the approximations $\tilde{f}_I$ are usually of lower degree than $f_I$. As a consequence, the cost at such an interval is tremendously reduced because we have to compute the polynomial $f_{I^+}=f_I(\frac{1}{4n}+(1+\frac{1}{2n})x)$ which is expensive for large degrees. In particular, for a polynomial with two very nearby roots (such as Mignotte polynomials), this behavior can be clearly observed. More precisely, when refining an interval $I$ which contains two nearby roots, the degree of $\tilde{f}_I$ decreases in each bisection step and eventually equals $2$ for $I$ small enough. We consider this behavior as quite natural because $f_I$ implicitly captures the information on the location of the roots in a neighborhood of $I$ whereas the influence of all other roots becomes almost negligible. In comparison to previous methods, the proposed algorithm exploits this fact.

\section{Conclusion}\label{conclusion}

We presented a new complete and deterministic algorithm to isolate the real 
roots of an arbitrary square-free polynomial $F$ with real coefficients. Our analysis shows that the hardness of isolating the real roots exclusively depends on the 
location of the roots and not on the coefficient type. Furthermore, the overall running time is significantly reduced by considering approximations at each node of the recursion tree. In particular, for integer polynomials, we achieve an improvement with respect to worst case bit complexity by a factor $n=\deg F$ compared to the best bounds known for other practical methods such as the Descartes or the continued fraction method. The latter is due to the fact that exact arithmetic produces too much 
information for the task of root isolation and, thus, a significant 
overhead of computation. We remark that recent work~\cite{ks-qir-issac11} shows corresponding results for the task of further refining given isolating intervals.
Since we are aiming for a practical method, we 
formulated our algorithm in the spirit of the Descartes method. We are convinced that because of its similarities to the latter exact 
method and because of its usage of approximate and less expensive 
computations, it will prove to be efficient in practice as well. We plan 
to implement our algorithm to verify this claim. A first promising step~\cite{st:polyalgebraic:11} has already been made.
Finally, univariate root isolation constitutes an important substep in cad (cylindrical algebraic decomposition) computations. In combination with a recent result on the complexity for real root approximation~\cite{ks-qir-issac11}, it is possible to obtain a bound on the worst case bit complexity of topology computation of planar algebraic curves~\cite{ks-topology-jsc11} which crucially improves upon the existing record bounds from~\cite{det-jsc,kerber-phd}.

\bibliography{localref,ref,yap,exact,alge}
\bibliographystyle{abbrv}

\newpage
\section{Appendix}

\subsection{Approximating $\Gamma_F$}\label{sec:rootbound}

In this section, we show how to compute an integer approximation $\Gamma\in\N$ of the exact logarithmic root bound $\Gamma_F:=\log(\max_i|\xi_i|)$ with $\Gamma_F\le\Gamma<4\log n+\Gamma_F$. We further prove that this computation can be done with $\tilde{O}((n\Gamma_F)^2)$ bit operations. As a byproduct,  $\tau=\tilde{O}(n\Gamma_F)$.

Consider the Cauchy polynomial $F^C(x):=|A_n|x^n-\sum_{i=0}^{n-1}|A_i|x^i$ of $F$. Then, $F^C$ has a unique positive real root $\xi^*\in\R^+$ and it holds that $$2^{\Gamma_F}\le\xi^*<\frac{n}{\ln 2}\cdot 2^{\Gamma_F}<2n\cdot 2^{\Gamma_F};$$
see~\cite[Proposition 2.51]{eigenwillig:thesis}. Thus, it follows that $F^C(x)>0$ for all $x\ge n2^{\Gamma_F+1}$ and, in particular, $|A_n|(n2^{\Gamma_F+1})^n>|A_i|(n2^{\Gamma_F+1})^i$ for all $i$. From the definition of $\tau$ on page~\pageref{deftau}, we have $\tau=1$ or there exists an $i_0$ with $|A_{i_0}|/|A_n|\ge 2^{\tau-1}$. The first case is trivial and, in the second case, we have $(2^{\Gamma_F+1}n)^{n-i_0}\ge 2^{\tau-1}$. Thus, $\tau=\tilde{O}(n\Gamma_F)$.

We now aim to compute an approximation of $\xi^*$ via evaluating $F^C$ at $x=2,4,8,\ldots$. Then, for the smallest $k\in\N$ (denoted by $k_0$) with $F^C(2^k)>0$, we must have $\xi^*<2^{k_0}<2\xi^*$. However, since $F$ has approximate coefficients, these evaluations cannot be done exactly. The idea is to use interval arithmetic with a certain precision $\rho$ (fixed point arithmetic) such that $\IBox(F^C(2^k),\rho)<1$, where $\IBox(E,\rho)$ is the interval obtained by evaluating of a polynomial expression $E$ via interval arithmetic with precision $2^{-\rho}$ for the basic arithmetic operations; see~\cite[Section 4]{ks-qir-issac11} for details. We initially start with $k=1$. If $\IBox(F^C(2^k),\rho)$ contains zero or $\lambda<0$ for all $\lambda\in\IBox(F^C(2^k),\rho)$, we proceed with $k+1$. Otherwise, we must have $k_0\le k\le k_0+1$: The left inequality is obvious. For the right inequality, we remark that $2^{k_0+1}$ has distance larger than $1$ to all roots of $F^C$ and, thus, $F^C(2^{k_0+1})\ge 1$. Hence, $\lambda>0$ for all $\lambda\in\IBox(F^C(2^{k_0+1}),\rho)$. It follows that $\Gamma:=k$ fulfills $$\Gamma_F\le\Gamma\le k_0+1<\log(4\xi^*)<\log(8n2^{\Gamma_F})<4\log n+\Gamma_F.$$

It remains to bound the cost for the interval computations of $\IBox(F^C(2^k),\rho)$. Since $F^C$ has coefficients of size less than $2^\tau$, we have to choose $\rho$ such that 
\[
2^{-\rho+2}(n+1)^2 2^{\tau+nk}<1
\] 
in order to guarantee that $w(\IBox(F^C(2^k),\rho))<1$; see~\cite[Lemma 3]{ks-qir-issac11}. Then, $\rho$ is bounded by $O(\tau+nk)$ and, thus, each interval evaluation needs $\tilde{O}(n(\tau+nk))$ bit operations. From our above considerations, we have $k\le k_0+1=O(\log n+\Gamma_F)$, hence, the total cost is bounded by $\tilde{O}(n\Gamma_F(\tau+n\Gamma_F))=\tilde{O}((n\Gamma_F)^2)$ and we need approximations of $F$ to $O(\tau+n\Gamma_F)=\tilde{O}(n\Gamma_F)$ bits after the binary point.

\subsection{Integer Polynomials}\label{integer}

For an integer polynomial $F\in\Z[x]$ as given in (\ref{polyF}), we aim to show that $\Sigma_F=\Otilde(n\tau)$. We proceed in two steps: First, we cluster the roots $\xi_i$ of $F$ into subsets consisting of nearby roots. Second, we apply the generalized Davenport-Mahler bound~\cite{du-sharma-yap:sturm:07,eigenwillig:thesis} to the roots of $F$. Eventually, the above result follows.\\

W.l.o.g., we can assume that $\sigma(\xi_1,F)\le\ldots,\le\sigma(\xi_n,F)$. For $h\in\N$, we denote $i(h)$ the maximal index $i$ with $\sigma(\xi_i,F)\le 2^{-h}$ and $R=R(h):=\{\xi_1,\ldots,\xi_{i(h)}\}$ the corresponding set of roots $\xi_i$ with $\sigma(\xi_i,F)\le 2^{-h}$. If $h\le\log(1/\sigma_F)$, then $R$ contains at least two roots. We are interested in a partition of $R$ into disjoint subsets $R_1,\ldots,R_l$ that consist of nearby points, only. 

\begin{lem}\label{cluster}
Suppose that $h\le\log(1/\sigma_F)$. Then, there exists a partition of 
$R:=R(h)$ into disjoint sets $R_1,\ldots,R_l$ such that $|R_{i}|\geq 2$ for 
all $i\in\{1,\ldots,l\}$ and $|\xi-\xi'|\le n2^{-h}$ for all $\xi$, $\xi'\in 
R_{i}$.
\end{lem}
\begin{proof}
We initially set $R_1:=\{\xi_1\}$. Then, we add all roots $\xi_i$ to $R_1$ 
that satisfy $|\xi_i-\xi_1|\leq 2^{-h}$. For each root in 
$R_1$, we proceed in the same way. More precisely, for each $\xi\in R_1$, we add those roots $\xi'\in R$ to $R_1$ with $|\xi-\xi'|\le 2^{-h}$. If no further root can be added to $R_1$, we consider 
the set $R\backslash R_1$ of the remaining roots and treat it in 
exactly the same manner. Finally, we end up with a partition 
$R_{1},\ldots,R_{l}$ of $R$ such that, for any two points in any $R_{i}$, their 
distance is less than or equal to $(|R_{i}|-1)2^{-h}<n2^{-h}$. Furthermore, each of 
the sets $R_i$ must contain at least two roots as $\sigma(\xi_i,F)\leq 2^{-h}$ for all $i=1,\ldots,i(h)$. 
\end{proof}

We now consider a directed graph $\mathcal{G}_i$ on each $R_i$
which connects consecutive roots of $R_i$ in ascending order of their
absolute values. We define $\mathcal{G}:=(R,E)$ as the union of all
$\mathcal{G}_i$. Then, $\mathcal{G}$ is a directed graph on $R$
with the following properties:
\begin{enumerate}
\item each edge $(\alpha,\beta)\in E$ satisfies $|\alpha|\leq |\beta|$,
\item $\mathcal{G}$ is acyclic, and
\item the in-degree of any node is at most 1.
\end{enumerate} 
Hence, we can apply the generalized Davenport-Mahler
bound~\cite{du-sharma-yap:sturm:07,eigenwillig:thesis}
to $\mathcal{G}$:
\[
	\prod_{(\alpha,\beta)\in E}|\alpha-\beta|
		\geq\frac{1}{(\sqrt{n+1}2^{\tau})^{n-1}}
		\cdot\left(\frac{\sqrt{3}}{n}\right)^{\# E}
		\cdot\left(\frac{1}{n}\right)^{n/2}
\]
As each set $R_i$ contains at least $2$ roots,
we must have $i(h)>\# E\geq i(h)/2$.
Furthermore, for each edge $(\alpha,\beta)\in E$, we have
$|\alpha-\beta|\leq n2^{-h}$. It follows that
	\begin{align*}
	(n2^{-h})^{\frac{i(h)}{2}}&
		> \frac{1}{(\sqrt{n+1}2^{\tau})^{n-1}}
		 \cdot\left(\frac{\sqrt{3}}{n}\right)^{i(h)}
		\cdot\left(\frac{1}{n}\right)^{n/2}
	 > \frac{1}{(n+1)^n2^{n\tau}}
		\cdot\left(\frac{3}{n^2}\right)^{i(h)/2}
	\end{align*}
and, thus,
	$$
	i(h)<\frac{2n(\tau+\log(n+1))}{\log 3+\log n+h}<\frac{2n(\tau+\log(n+1))}{h}.
	$$
It directly follows that $\log(1/\sigma_F)<n(\tau+\log(n+1))+1$ since, otherwise, there would exist an 
$h$ with $n(\tau+\log(n+1))< h\le\log(1/\sigma_F)$ and $i(h)<2$ which is not possible. For the bound on $\Sigma_F$, it suffices to consider only the roots $\xi_1,\ldots,\xi_k$ with separation $\le 1/2$ since all other roots contribute with at most $n$ to the sum $\Sigma_F$. Since
$$
-\sum_{i=1}^k \log\sigma(\xi,F)<\sum_{h=1}^{\left\lceil n(\tau+\log(n+1))\right\rceil} i(h)<2n(\tau+\log(n+1))\sum_{h=1}^{\left\lceil n(\tau+\log(n+1))\right\rceil}\frac{1}{h}=O(n\tau\log(n\tau)),
$$ 
it follows that $\Sigma_F=\Otilde(n\tau)$.

\newpage

\subsection{Algorithms}\label{appendix:algos}

\begin{algorithm}
\caption{$\dcm$}
\label{alg:dcm}
\begin{algorithmic}
\REQUIRE {polynomial $f =\sum_{0 \le i \le n} a_i x^i\in\R[x]$ as defined in (\ref{polyf})}\medskip
\ENSURE{returns a list $\mathcal{O}$ of disjoint isolating intervals for all real roots of $f$}

\hfill \COMMENT{only in the REAL-RAM model}
\STATE $I_0 \assign (-\frac{1}{2},\frac{1}{2})$
\STATE $f_{I_0}(x)\assign f(-\frac{1}{2}+x)$ 
\STATE $\mathcal{A} \assign \sset{(I_0,f_{I_0})}$; $\mathcal{O} \assign \emptyset$ \hfill \COMMENT{list of
active and isolating intervals}
\REPEAT
\STATE $(I,f_I)$ some element in $\mathcal{A}$ with $I=(a,b)$; 
delete $(I,f_I)$ from $\mathcal{A}$
\STATE $f_{I^+}\assign f_I\left(-\frac{1}{4n}+\left(1+\frac{1}{2n}\right)x\right)$ and 
$f_{I^+}^t(x)=\sum_{i=0}^n h_i x^i\assign (1+x)^n\cdot 
f_{I^+}\left(\frac{1}{1+x}\right)$
\IF {$\var(f_{I^+}^t)=0$}
\STATE do nothing
\ELSE 
\IF {$t_{3/2}^{(f_{I})'}(0,2)>0$}
\STATE $s\assign \sgn f_{I^{+}}(0)\cdot f_{I^{+}}(1)$
\IF {$s\ge 0$}
\STATE do nothing
\ELSE
\IF {$I^+$ does not intersect any interval in $\mathcal{O}$}
\STATE add $I^+$ to $\mathcal{O}$
\ELSE
\STATE do nothing
\ENDIF
\ENDIF
\ELSE
\STATE subdivide $I$ into $I_l\assign (a,m_I)$ and $I_r\assign (m_I,b)$
\STATE $f_{I_l}\assign f_I\left(\frac{x}{2}\right)$ and $f_{I_r}\assign f_I\left(\frac{x+1}{2}\right)=f_{I_l}(x+1)$
\STATE add $(I_l,f_{I_l})$ and $(I_r,f_{I_r})$ to $\mathcal{A}$
\ENDIF
\ENDIF
\UNTIL{$\mathcal{A}$ is empty}
\RETURN $\mathcal{O}$
\end{algorithmic}
\end{algorithm}\vspace{-0.2cm}

\begin{algorithm}[t]
\caption{$\dcm^\rho$}
\label{alg:dcml}
\begin{algorithmic}
\REQUIRE {polynomial $f =\sum_{0 \le i \le n} a_i x^i\in\R[x]$ as in (\ref{polyf}) and a $\rho\in\N$}\medskip
\ENSURE{returns "insufficient precision" or a list $\mathcal{O}=\{J_k,s_{k,l},s_{k,r},B_k\}$ of disjoint isolating intervals $J_k=(c_k,d_k)$ for some of the real roots of $f$ (and $s_{k,l}=\sgn f(c_k)$, $s_{k,r}=\sgn f(d_k)$ and $0<B_k\le\min(|f(c_k)|,|f(d_k)|)$.}\medskip
\STATE $I_0 \assign (-\frac{1}{2},\frac{1}{2})$
\STATE $\tilde{f}$ a $(\rho+n+1)$-binary approximation of $f$
\STATE $\tilde{f_{I_0}}$ a $(\rho+1)$-binary approximation of 
$\tilde{f}(-\frac{1}{2}+x)$ \hfill\COMMENT{$\Rightarrow\tilde{f_{I_0}}\in [f_{I_0}]_{2^{-\rho}}$ }
\STATE $\mathcal{A} \assign \sset{(I_0,\tilde{f_{I_0}},\rho)}$; $\mathcal{O} \assign \emptyset$ \hfill \COMMENT{list of
active and isolating intervals}\vspace{0.5cm}

\REPEAT
\STATE $(I,\tilde{f_I},\rho_I)$, where $I \assign (a,b)$, some element in $\mathcal{A}$; delete $(I,\tilde{f_I},\rho_I)$ from $\mathcal{A}$
\STATE $\tilde{f}_{I^+}(x)\assign \tilde{f_I}\left(-\frac{1}{4n}+\left(1+\frac{1}{2n}\right) x\right)$ and
$\tilde{h}(x)=\sum_{i=0}^n \tilde{h}_i x^i:=(1+x)^n\cdot \tilde{f}_{I^+}^t\left(\frac{1}{1+x}\right)$
   \IF {$\tilde{h}_i>-2^{n+2-\rho_I}$ for all $i$ or $\tilde{h}_i<2^{n+2-\rho_I}$ for all $i$} \STATE do nothing
\ELSE 
       \IF {$t_{3/2}^{(\tilde{f_I})'}>-n2^{n+1-\rho_I}$} 
\STATE $\hat{f_I}(x)\assign \tilde{f_I}(x)+n2^{n+1-\rho_I}\cdot x$
\STATE $\lambda^-\assign \tilde{f}_{I^+}(0)-2^{n-1-\rho_I}$, $\lambda^+\assign \tilde{f}_{I^+}(1)+(4n+1)2^{n-1-\rho_I}$ and $\lambda\assign \tilde{f_I}(-1/n)-2^{n+1-\rho_I}$.
\IF {$\tilde{I}=(a-\frac{w(I)}{2n},b+\frac{w(I)}{2n}$ intersects no interval $J$ for all $(J,s_{J,l},s_{J,r},B_J)\in\mathcal{O}$ \textbf{and} $\lambda^-\cdot\lambda^+<0$ \textbf{and} $\min(|\lambda^-|,|\lambda^+|)>n2^{n+3-\rho_I}$ \textbf{and} $|\lambda|>n^2 2^{\deg(\hat{f_I})+7+n-\rho_I}$}
\STATE add $(\tilde{I},\sgn(\lambda^-),\sgn(\lambda^+),\min(|\lambda^-|,|\lambda^+|)-2^{n+3-\rho_I}n)$ to $\mathcal{O}$

\hfill \COMMENT{$\Rightarrow$ $\tilde{I}$ contains a root $\xi$ of $f$ and the $\frac{w(I)}{n}$-neighborhood of $I$ is isolating for $\xi$}
\ELSE
\STATE do nothing \hfill \COMMENT{$\tilde{J}$ is already isolating for $\xi$}
\ENDIF
\ELSE
\STATE do nothing
\ENDIF
\ELSE
\IF {$\rho_I<0$}
\RETURN "insufficient precision"
\ELSE
\IF {$\rho_I<2$}
\RETURN "insufficient precision"
\ELSE
\STATE Subdivide $I$ into $I_l\assign (a,m_I)$ and $I_r\assign(m_I,b)$
\STATE $\tilde{f_{I_l}}$ an $\rho_I$-binary approximation of $\tilde{f_I}\left(\frac{x}{2}\right)$\hfill\COMMENT{$\Rightarrow\tilde{f_{I_l}}\in[f_{I_l}]_{2^{-(\rho_I-1)}}$}
\STATE  $\tilde{f_{I_r}}$ an $(\rho_I-1)$-binary approximation of $\tilde{f_I}\left(\frac{1+x}{2}\right)$\hfill\COMMENT{$\Rightarrow\tilde{f_{I_r}}\in[f_{I_r}]_{2^{-(\rho_I-2)}}$}
\STATE Add $(I_l,\tilde{f_{I_l}},\rho_I-1)$ and $(I_r,\tilde{f_{I_r}},\rho_I-2)$ to $\mathcal{A}$
\ENDIF
\ENDIF
\ENDIF
\UNTIL{$\mathcal{A}$ is empty}
\RETURN $\mathcal{O}$
\end{algorithmic}
\end{algorithm}

\begin{algorithm}[t]
\caption{$\textsc{Certify}^\rho$}
\label{alg:certify}
\begin{algorithmic}
\REQUIRE {polynomial $f =\sum_{0 \le i \le n} a_i x^i\in\R[x]$ as defined in (\ref{polyf}), an $\rho\in\N$ and the list $\mathcal{O}=\{(J_k,s_{k,l},s_{k,r},B_k)\}_{k=1,\ldots,s}$ returned by $\dcm^\rho$.}\medskip
\ENSURE{returns "insufficient precision" or the list $\mathcal{L}=\{J_k\}_{k=1,\ldots,s}$ of isolating intervals with the guarantee that, for each real root of $f$, there exists a corresponding interval in $\mathcal{L}$.}\medskip
\STATE $I_0 \assign (-\frac{1}{2},\frac{1}{2})$
\STATE $\tilde{f}$ an $(\rho+n+1)$-binary approximation of $f$
\STATE $\tilde{f_{I_0}}$ an $(\rho+1)$-binary approximation of $\tilde{f}(-\frac{1}{2}+x)$ \hfill\COMMENT{$\Rightarrow\tilde{f_{I_0}}\in [f_{I_0}]_{2^{-\rho}}$ }
\STATE $\mathcal{A} \assign \sset{(I_0,\tilde{f_{I_0}},\rho)}$ \hfill \COMMENT{list of
active intervals}\vspace{0.3cm}
\REPEAT
\STATE $(I,\tilde{f_I},\rho_I)$, where $I \assign (a,b)$, some element in $\mathcal{A}$; delete $(I,\tilde{f_I},\rho_I)$ from $\mathcal{A}$.
\IF {$\bar{I}\cap \mathcal{R}=\bigcup_{i=1}^s L_i=\emptyset$}
\STATE do nothing
\ELSE
\IF {$t_{3/2}^{\tilde{f_I}}(0,1)>-(n+1)2^{-\rho_I+1}$}
\IF {$|\tilde{f_I}(0)+2^{-\rho_I+2}n|>(n+1)2^{-\rho_I+5}$}
\STATE do nothing \hfill\COMMENT{$\overline{I}$ contains no root of $f$}
\ELSE
\RETURN "insufficient precision"\hfill\COMMENT{$\rho<\rho_f^{\max}$}
\ENDIF
\ELSE
\STATE $\tilde{h}(x)\assign \sum_{i=0}^n \tilde{h}_i x^i= (1+x)^n(\tilde{f_I})'(\frac{1}{1+x})$
\IF {$\tilde{h}_i<n2^{n-\rho_I}$ for all $i$ (or $\tilde{h}_i>-n2^{n-\rho_I}$ for all $i$)}
\STATE $g(x):=\tilde{f}_I(x)-n2^{n-\rho_I}$ (or $g(x):=\tilde{f}_I(x)+n2^{n-\rho_I}$, respectively);
\IF {for each $L_i=[q_l,q_r]$, $\min(|\lambda(q_l)|,|\lambda(q_r)|)>n2^{n+2-\rho_I}$\textbf{ and }$\lambda(q_l)\cdot\lambda(q_r)<0$}
\STATE do nothing
\hfill\COMMENT{$\overline{I}\cap\mathcal{R}$ contains no root of $f$; $\lambda(q_l)$, $\lambda(q_r)$ defined as in (\ref{lambdaq})}
\ELSE
\RETURN "insufficient precision"\hfill\COMMENT{$\rho<\rho_f^{\max}$}
\ENDIF
\ELSE
\IF {$\rho_I<2$}
\RETURN "insufficient precision"
\ELSE
\STATE Subdivide $I$ into $I_l\assign (a,m_I)$ and $I_r\assign(m_I,b)$
\STATE $\tilde{f_{I_l}}$ an $\rho_I$-binary approximation of $\tilde{f_I}\left(\frac{x}{2}\right)$\hfill\COMMENT{$\Rightarrow\tilde{f_{I_l}}\in[f_{I_l}]_{2^{-(\rho_I-1)}}$}
\STATE  $\tilde{f_{I_r}}$ an $(\rho_I-1)$-binary approximation of $\tilde{f_I}\left(\frac{1+x}{2}\right)$\hfill\COMMENT{$\Rightarrow\tilde{f_{I_r}}\in[f_{I_r}]_{2^{-(\rho_I-2)}}$}
\STATE Add $(I_l,\tilde{f_{I_l}},\rho_I-1)$ and $(I_r,\tilde{f_{I_r}},\rho_I-2)$ to $\mathcal{A}$
\ENDIF 
\ENDIF
\ENDIF
\ENDIF
\UNTIL{$\mathcal{A}$ is empty}
\RETURN "certification successful" \hfill\COMMENT{The region of uncertainty $\mathcal{R}$ contains no root of $f$}
\end{algorithmic}
\end{algorithm}

\end{document}